\documentclass[11pt,journal]{IEEEtran}%
\usepackage{latexsym,delarray,multicol}
\usepackage{epsfig}
\usepackage{amsfonts}
\usepackage{mathrsfs}
\usepackage{amsthm}
\usepackage{amsmath}
\usepackage{amssymb}
\usepackage{revsymb}
\usepackage{graphicx}
\usepackage{color}
\usepackage{hyperref}%

\newtheorem{corollary}{Corollary}
\newtheorem{definition}{Definition}
\newtheorem{proposition}{Proposition}
\newtheorem{lemma}{Lemma}
\onecolumn

\begin{document}

\title{Distance verification for classical and quantum LDPC codes}
\date{\today}
\author{Ilya Dumer,~\thanks{Ilya Dumer is with the Department of Electrical and
Computer Engineering, University of California, Riverside, California, 92521,
USA; (e-mail: dumer@ee.ucr.edu).} Alexey A. Kovalev,\thanks{Alexey A. Kovalev
is with the Department of Physics and Astronomy, University of
Nebraska---Lincoln, USA; (e-mail: alexey.kovalev@unl.edu). Research supported
in part by the NSF under Grant No. PHY-1415600.} and Leonid P.~Pryadko\thanks{
Leonid P.~Pryadko is with the Department of Physics \& Astronomy, University
of California, Riverside, California, 92521, USA (e-mail: leonid@ucr.edu).
Research supported in part by the U.S. Army Research Office under Grant No.
W911NF-14-1-0272 and by the NSF under Grant No. PHY-1416578.} \thanks{The
material of this paper was presented in part at the 2014 IEEE Symp. Info.
Theory, Honolulu, HI, USA, July 1-6, 2014 and at the 2016 IEEE Symp. Info.
Theory, Barcelona, Spain, USA, July 4-8, 2016.}}
\maketitle

\begin{abstract}
The techniques of distance verification known for general linear codes are
re-applied to quantum stabilizer codes. Then distance verification is
addressed for classical and quantum LDPC codes. New complexity bounds for
distance verification with provable performance are derived using the average
weight spectra of the ensembles of LDPC codes. These bounds are expressed in
terms of the erasure-correcting capacity of the corresponding ensemble. We
also present a new irreducible-cluster technique that can be applied to any
LDPC code and takes advantage of parity-checks' sparsity for both classical
and quantum LDPC codes. This technique reduces complexity exponents of all
existing deterministic techniques designed for generic stabilizer codes with
small relative distances, which also include all known families of quantum
LDPC codes.

\textbf{Index Terms} -- Distance verification, complexity bounds, quantum
stabilizer codes, LDPC codes, erasure correction \vspace{-0.15in}

\end{abstract}

\vspace{-0.15in}

\section{Introduction}

Quantum error correction (QEC)
\cite{shor-error-correct,Knill-Laflamme-1997,Bennett-1996} is a critical part
of quantum computing due to the fragility of quantum states. Two related code
families, surface (toric) quantum codes
\cite{kitaev-anyons,Dennis-Kitaev-Landahl-Preskill-2002} and topological color
codes \cite{Bombin-MartinDelgado-2006,Bombin-2007, Bombin-MartinDelgado-2007},
have been of particular interest in quantum design
\cite{Bombin-MartinDelgado-2007}, \cite{Raussendorf-Harrington-2007}. Firstly,
these codes only require simple local gates for quantum syndrome measurements.
Secondly, they efficiently correct some non-vanishing fraction of errors,
below a fault-tolerant threshold of about 1\% per gate. Unfortunately,
locality limits such codes to an asymptotically zero code rate
\cite{Bravyi-Poulin-Terhal-2010} and makes a useful quantum computer
prohibitively large. Therefore, there is much interest in feasible quantum
coding with no local restrictions.

Low-density-parity-check (LDPC) codes
\cite{Postol-2001,MacKay-Mitchison-McFadden-2004} form a more general class of
quantum codes. These codes assume no locality but only require low-weight
stabilizer generators (parity checks). Unlike locally-restricted codes, they
also achieve a finite code rate along with a non-zero error probability
threshold, both in the standard setting, and in a fault-tolerant setting, when
syndrome measurements include errors
\cite{Kovalev-Pryadko-FT-2013,Dumer-Kovalev-Pryadko-bnd-2015}. However,
quantum LDPC codes are still much inferior to their classical counterparts.
Namely, all existing quantum LDPC codes with bounded stabilizer weight
\cite{Tillich-Zemor-2009,Zemor-2009,Couvreur-Delfosse-Zemor-2012,
Kovalev-Pryadko-2012,Andriyanova-Maurice-Tillich-2012,
Kovalev-Pryadko-Hyperbicycle-2013,
Bravyi-Hastings-2013,Guth-Lubotzky-2014,Freedman-Meyer-Luo-2002} have code
distances $d$ that scale at most as $\sqrt{n\ln n}$ in length $n$, unlike
linear scaling in the classical LDPC codes. Many of the existing quantum
constructions also exhibit substantial gaps between the upper and lower bounds
for their distances $d$. In particular, the recent quantum design of
\cite{Bravyi-Hastings-2013} yields the orders of $n$ and $\sqrt{n}$ for these
bounds. Finding the exact distances of such codes is thus an important open problem.

This paper addresses various numerical algorithms that verify code distance
with provable performance for the classical LDPC codes, quantum stabilizer
codes, and quantum LDPC codes. Given some ensemble of codes, we wish to verify
code distances for most codes in this ensemble with an infinitesimal
probability of failure. In particular, we will discuss deterministic
algorithms that yield no failures for most codes in a given ensemble. We also
address probabilistic algorithms that have a vanishing probability of failure.
This high-fidelity setting immediately raises important complexity issues.
Indeed, finding the code distance of a generic code is an NP-hard problem.
This is valid for both the exact setting \cite{Vardy-1997} and the evaluation
problem \cite{Dumer-2003,Cheng-2009}, where we only verify if $d$ belongs to
some interval $[\delta,c\delta]$ for a given constant $c\in(1,2)$. In this
regard, we note that all algorithms discussed below still have exponential
complexity in block length $n,$ if the average code distance grows linearly in
a given ensemble. Below, we consider both binary and $q$-ary codes and wish to
achieve the lowest exponential complexity $q^{Fn}$ for distance verification
of classical or quantum LDPC codes.

We analyze complexity exponents $F$ in three steps. Section \ref{sec:generic}
establishes a framework for generic quantum codes. To do so, we revisit
several algorithms known for classical linear codes. Then we re-apply these
techniques for quantum stabilizer codes. Given the complexity benchmarks of
Section \ref{sec:generic}, we then address binary LDPC codes in Section
\ref{sec:LDPC}. Here we can no longer use the generic properties of random
generator (or parity-check) matrices. Therefore, we modify the existing
algorithms to include the LDPC setting. In particular, we show that only a
vanishing fraction of codes may have atypically high complexity. These codes
are then discarded. As a result, we re-define known complexity estimates in
terms of two parameters: the average code distance and the erasure-correcting
capacity of a specific code ensemble. To estimate this capacity, we use the
average weight spectra, which were derived in \cite{Gallager} for the original
ensemble of LDPC codes and in \cite{Litsyn-2002} for a few other LDPC
ensembles. Our complexity estimates hold for any ensemble given its
erasure-correcting capacity or some lower bound. More generally, these
algorithms perform list decoding within distance $d$ from any received vector
$y,$ whereas distance verification does so for $y=0.$

Here, however, we leave out some efficient algorithms that require more
specific estimates. In particular, we do not address belief propagation (BP)
algorithms, which can erroneously end when they meet stopping sets, and
therefore fail to furnish distance verification with an arbitrarily high
likelihood. Despite this, the simulation results presented in  papers
\cite{Declercq-Fossorier-2008} and \cite{Hu-Fossorier-Eleftheriou-ISIT-2004}
show that list decoding BP algorithms can also be effective in distance verification.

In Section \ref{sec:LC-LDPC}, we consider quantum stabilizer LDPC codes. These
codes use some self-orthogonal quaternary code $C$ and its dual $C^{\bot}$.
This self-orthogonality separates quantum LDPC codes from their conventional
counterparts. One particular difference is a low relative distance of the
existing constructions, the other is a substantial number of short cycles in
their graphical representation. The latter fact also complicates BP
algorithms. For these reasons, our goal is to design new algorithms that are
valid for any LDPC code including any quantum code. To do so, we use the fact
that verification algorithms may seek only irreducible
\cite{Dumer-Kovalev-Pryadko-bnd-2015} codewords that cannot be separated into
two or more non-overlapping codewords. This approach yields a cluster-based
algorithm that exponentially reduces the complexity of all known deterministic
techniques for sufficiently small relative distance $d/n$, which is the case
for the existing families of quantum LDPC codes. This algorithm also
generalizes the algorithm of \cite{Dumer-Kovalev-Pryadko-bnd-2015} for
nonbinary LDPC codes.

Consider a $q$-ary $(\ell,m)$-\emph{regular} LDPC code, which has $\ell$
non-zero symbols in each column and $m$ non-zero symbols in each row of its
parity-check matrix. Let $h_{2}(x)$ be the binary entropy of $x\in
\lbrack0,1].$ Our main results are presented in Propositions \ref{prop:CS-1}
and \ref{prop:IC-1} and can be summarized as follows.

\begin{proposition}
\label{prop:1-1}Consider any permutation-invariant ensemble $\mathbb{C}$ of
$q$-ary linear codes with relative distance {$\delta_{\ast}.$ Let }%
$\theta_{\ast}$ denote the expected erasure-correcting capacity for codes
$C\in\mathbb{C}$. For most codes $C\in\mathbb{C}$, the code distance
{$\delta_{\ast}n$ }can be verified with complexity of order $2^{Fn}$, where
$F=h_{2}(\delta_{\ast})-{\theta}_{\ast}h_{2}(\delta_{\ast}/{\theta}_{\ast}).$
For any $q$-ary $(\ell,m)$-regular LDPC code (classical or quantum), the code
distance {$\delta_{\ast}n$ }can be verified with complexity of order $2^{Fn}$,
where $F=\delta_{\ast}\log_{2}(\gamma_{m}(m-1))$ and $\gamma_{m}$ grows
monotonically with $m$ in the interval $(1,\left(  q-1\right)  /\ln q).$
\end{proposition}

\section{Background}

Let $C[n,k]_{q}$ be a $q$-ary linear code of length $n$ and dimension $k$ in
the vector space $\mathbb{F}_{q}^{n}$ over the field $\mathbb{F}_{q}$. This
code is specified by the parity check matrix $H$, namely $C=\{\mathbf{c}%
\in\mathbb{F}_{q}^{n}|H\mathbf{c}=0\}$. Let $d$ denote the Hamming distance of
code $C$.

A quantum $[[n,k]]$ stabilizer code $Q$ is a $2^{k}$-dimensional subspace of
the $n$-qubit Hilbert space $\mathbb{H}_{2}^{\otimes n}$, a common $+1$
eigenspace of all operators in an Abelian \emph{stabilizer group}
$\mathscr{S}\subset\mathscr{P}_{n}$, $-\openone\not \in \mathscr{S}$, where
the $n$-qubit Pauli group $\mathscr{P}_{n}$ is generated by tensor products of
the $X$ and $Z$ single-qubit Pauli operators. The stabilizer is typically
specified in terms of its generators, $\mathscr{S}=\left\langle S_{1}%
,\ldots,S_{n-k}\right\rangle $; measuring the generators $S_{i}$ produces the
\emph{syndrome} vector. The weight of a Pauli operator is the number of qubits
it affects. The distance $d$ of a quantum code is the minimum weight of an
operator $U$ which commutes with all operators from the stabilizer
$\mathscr{S}$, but is not a part of the stabilizer, $U\not \in \mathscr{S}$.

A Pauli operator $U\equiv i^{m}X^{\mathbf{v}}Z^{\mathbf{u}}$, where
$\mathbf{v},\mathbf{u}\in\{0,1\}^{\otimes n}$ and $X^{\mathbf{v}}=X_{1}%
^{v_{1}}X_{2}^{v_{2}}\ldots X_{n}^{v_{n}}$, $Z^{\mathbf{u}}=Z_{1}^{u_{1}}%
Z_{2}^{u_{2}}\ldots Z_{n}^{u_{n}}$, can be mapped, up to a phase, to a
quaternary vector, $\mathbf{e}\equiv\mathbf{u}+\omega\mathbf{v}$, where
$\omega^{2}\equiv\overline{\omega}\equiv\omega+1$. A product of two quantum
operators corresponds to the sum ($\bmod\,2$) of the corresponding vectors.
Two Pauli operators commute if and only if the \emph{trace inner product}
$\mathbf{e}_{1}\ast\mathbf{e}_{2}\equiv\mathbf{e}_{1}\cdot\overline
{\mathbf{e}}_{2}+\overline{\mathbf{e}}_{1}\cdot\mathbf{e}_{2}$ of the
corresponding vectors is zero, where $\overline{\mathbf{e}}\equiv
\mathbf{u}+\overline{\omega}\mathbf{v}$. With this map, an $[[n,k]]$
stabilizer code $Q$ is defined by $n-k$ generators of a stabilizer group,
which generate some \emph{additive} self-orthogonal code $C$ of size $2^{n-k}$
over $\mathbb{F}_{4}$. \cite{Calderbank-1997}. The vectors of code $C$
correspond to stabilizer generators that act trivially on the code; these
vectors form the \emph{degeneracy group} and are omitted from the distance
calculation. For this reason, any stabilizer code $Q$ has a code distance
\cite{Calderbank-1997} that is defined by the minimum non-zero weight in the
code $C^{\bot}\setminus C$.

An LDPC code, quantum or classical, is a code with a sparse parity check
matrix. A huge advantage of classical LDPC codes is that they can be decoded
in linear time using iterative BP algorithms
\cite{Gallager-1962,MacKay-book-2003}. Unfortunately, this is not necessarily
the case for quantum LDPC codes, which have many short cycles of length four
in their Tanner graphs. In turn, these cycles cause a drastic deterioration in
the convergence of the BP algorithm \cite{Poulin-Chung-2008}. This problem can
be circumvented with specially designed quantum codes
\cite{Kasai-Hagiwara-Imai-Sakaniwa-2012,Andriyanova-Maurice-Tillich-2012}, but
a general solution is not known.

\section{Generic techniques for distance verification \label{sec:generic}}

The problem of verifying the distance $d$ of a linear code (finding a
minimum-weight codeword) is related to a more general list decoding problem:
find all or some codewords at distance $d$ from the received vector. As
mentioned above, the number of operations $N$ required for distance
verification can be usually defined by some positive exponent $F=\overline
{\lim}$ ($\log_{q}N)/n$ as $n\rightarrow\infty$. For a linear $q$-ary code
with $k$ information qubits, one basic decoding algorithm inspects all
$q^{Rn}$ distinct codewords, where $R=k/n$ is the code rate. Another basic
algorithm stores the list of all $q^{n-k}$ syndromes and coset leaders. This
setting gives (space) complexity $F=1-R$. We will now survey some techniques
that are known to reduce the exponent $F$ for linear codes and re-apply these
techniques for quantum codes. For classical codes, most results discussed
below are also extensively covered in the literature (including our citations
below). In particular, we refer to \cite{Barg-1998} for detailed proofs.

\subsection{Sliding window (SW) technique\label{sec:sliding}}

Consider ensemble $\mathbb{C}$ of linear codes $C[n,k]$ generated by the
randomly chosen $q$-ary $(Rn\times n)$ matrices $G$. It is well known that for
$n\rightarrow\infty,$ most codes in ensemble $\mathbb{C}$ have full dimension
$k=Rn$ and meet the asymptotic GV bound $R=1-h_{q}(d/n)$, where
\begin{equation}
h_{q}(x)=x\log_{q}(q-1)-x\log_{q}x-(1-x)\log_{q}(1-x)\label{eq:Hq}%
\end{equation}
is the $q$-ary entropy function. We use notation $c_{I}$ and $C_{I}$ for any
vector $c$ and any code $C$ punctured to some subset of positions $I.$
Consider a \emph{sliding window} $I(i,s)$, which is the set of $s$ cyclically
consecutive positions beginning with $i=0,\ldots,n-1$. \ It is easy to verify
that most random $q$-ary codes $C\in\mathbb{C}$ keep their full dimension $Rn$
on all $n$ subsets $I(i,s)$ of length $s=k+2\left\lfloor \log_{q}%
n\right\rfloor $. Let $\mathbb{C}_{s}$ be such a sub-ensemble of codes
$C\in\mathbb{C}$. Most codes $C\in\mathbb{C}_{s}$ also meet the GV bound,
since the remaining codes in $\mathbb{C\smallsetminus C}_{s}$ form a vanishing
fraction of ensemble $\mathbb{C}.$ Also, $\mathbb{C}_{s}$ includes all cyclic
codes$.$ We now consider the following SW technique of \cite{Evseev-1983}.

\begin{proposition}
\label{prop:SW}\cite{Evseev-1983} The code distance $\delta n$ of any linear
$q$-ary code $C[n,Rn]$ in the ensemble $\mathbb{C}_{s}$ can be found with
complexity $q^{nF_{\mathrm{C}}},$ where
\begin{equation}
F_{\mathrm{C}}=Rh_{q}(\delta)\label{sw0}%
\end{equation}
For most codes $C\in\mathbb{C}_{s},$ the complexity exponent is $F^{\ast
}=R(1-R).$
\end{proposition}

\begin{proof}
Given a code $C,$ we first verify if $C\in\mathbb{C}_{s},$ which requires
polynomial complexity. For such a code $C$, consider a codeword $c\in C$ of
weight $d=1,2,\ldots$. The weight of any vector $c_{I(i,s)}$ can change only
by one as $i+1$ replaces $i.$ Then some vector $c_{I(i,s)}$ of length $s$ has
the average Hamming weight $v\equiv\left\lfloor ds/n\right\rfloor $. Consider
all
\[
L=n(q-1)^{v}\textstyle{\binom{s}{v}}%
\]
vectors $c_{I(i,s)}$ of weight $v$ on each window $I(i,s)$. Then we use each
vector $c_{I(i,s)}$ as an information set and encode it to the full length
$n.$ The procedure stops if some encoded vector $c$ has weight $d$. This gives
the overall complexity $Ln^{2},$ which has the order of $q^{nF_{\mathrm{C}}}$
of (\ref{sw0}). For codes that meet the GV bound, this gives exponent
$F^{\ast}$.
\end{proof}

\textit{Remarks.} More generally, the encoding of vectors $c_{I(i,s)}$
represents erasure correction on the remaining $n-s$ positions. We use this
fact in the sequel for LDPC codes. Also, any error vector of weight $d$
generates vector $u$ of weight $v$ on some window $I=I(i,s)$. Thus, we can
subtract any error vector $u$ from the received vector $y_{I(i,s)}$ and
correct $d$ errors in code $C.$

We now proceed with ensemble $\mathbb{Q}$ of quantum stabilizer codes
$\mathcal{Q}$ $[[n,Rn]].$ Most of these codes meet the quantum GV bound
\cite{Ekert-Macchiavello-1996,Feng-Ma-2004}
\begin{equation}
R=1-2h_{4}(\delta)\label{eq:GV}%
\end{equation}
Any code $Q$ is defined by the corresponding additive quaternary code
$C^{\bot}$ and has the minimum distance $d(Q)=d(C^{\bot}\setminus C).$ Let
$\mathbb{Q}_{s}$ denote the ensemble of codes $Q,$ for which $C^{\bot}%
\in\mathbb{C}_{s}$. Note that $\mathbb{Q}_{s}$ includes most stabilizer codes.

\begin{corollary}
\label{cor:SW}The code distance $\delta n$ of any quantum stabilizer code $Q[[n,Rn]]$ in the ensemble $\mathbb{Q}_{s}$ can be found with complexity
$2^{nF_{\mathrm{SW}}},$ where
\begin{equation}
F_{\mathrm{SW}}=(1+R)h_{4}(\delta)\label{1}%
\end{equation}
For most codes in ensemble $\mathbb{Q}_{s},$ code distances $d$ can be found
with the complexity exponent
\begin{equation}
F_{\mathrm{SW}}^{\ast}=(1-R^{2})/2\label{eq:SW-GV-complexity}%
\end{equation}

\end{corollary}

\begin{proof}
For any quantum stabilizer code $Q[[n,k]],$ we apply the SW procedure to
the quaternary code $C^{\bot}$. Since code $C$ has size $2^{n-k}$ in the space
$\mathbb{F}_{4}^{n}$, its dual $C^{\bot}$ has the effective code rate
\footnote{This construction is analogous to pseudogenerators introduced in
Ref.~\cite{White-Grassl-2006}.}
\[
R^{\prime}=\left(  1-\textstyle\frac{n-k}{2n}\right)  =(1+R)/2
\]
which gives complexity $2^{nF_{\mathrm{SW}}}$ of (\ref{1}) for a generic
stabilizer code $Q$. \ Due to possible degeneracy, we also verify that any
encoded vector $c$ of weight $d$ does not belong to code $C$. Most generic
codes $Q$ $[[n,Rn]]$ also belong to ensemble $\mathbb{Q}_{s}$ and therefore
satisfy the quantum GV bound. The latter gives exponent
(\ref{eq:SW-GV-complexity}).
\end{proof}

Note that classical exponent $F^{\ast}=R(1-R)$ achieves its maximum $1/4$ at
$R=1/2$. By contrast, quantum exponent $F_{\mathrm{SW}}^{\ast}$ achieves its
maximum $1/2$ at the rate $R=0$.

\subsection{Matching Bipartition (MB) technique\label{sec:bipartition}}

\begin{proposition}
\label{prop:MB} The code distance $\delta n$ of any quantum stabilizer code
$Q[[n,Rn]]$ can be found with complexity $2^{nF_{\mathrm{MB}}},$ where%
\begin{equation}
F_{\mathrm{MB}}=h_{4}(\delta). \label{4}%
\end{equation}
For random stabilizer codes that meet the quantum GV bound (\ref{eq:GV}),
\begin{equation}
F_{\mathrm{MB}}^{\ast}=(1-R)/2. \label{5}%
\end{equation}

\end{proposition}

\begin{proof}
Similarly to the proof of Corollary \ref{cor:SW}, we consider any stabilizer
code $Q$ $[[n,Rn]]$ and the corresponding code $C^{\bot}$. \ For code
$C^{\bot},$ we now apply the algorithm of \cite{Dumer-1986, Dumer-1989}, which
uses two similar sliding windows, the \textquotedblleft left" window $I_{\ell
}(i,$ $s_{\ell})$ of length $s_{\ell}=\left\lfloor n/2\right\rfloor $ and the
complementary \textquotedblleft right" window $I_{r}$ of length $s_{r}%
=\left\lceil n/2\right\rceil $. For any vector $e$ of weight $d,$ consider
vectors $e_{\ell}$ and $e_{r}$ in windows $I_{\ell}$ and $I_{r}.$ At least one
choice of position $i$ then yields the average weights $v_{\ell}=\left\lfloor
d/2\right\rfloor $ and $v_{r}=\left\lceil d/2\right\rceil $ for both vectors.
For each $i,$ both sets $\{e_{\ell}\}$ and $\{e_{r}\}$ of such
\textquotedblleft average-weight\textquotedblright\ vectors have the size of
order $L=(q-1)^{d/2}{\binom{n/2}{d/2}}$.

We now calculate the syndromes of all vectors in sets $\{e_{\ell}\}$ and
$\{e_{r}\}$ to find matching vectors $(e_{\ell},e_{r})$, which give identical
syndromes, and form a codeword. Sorting the elements of the combined set
$\{e_{\ell}\}\cup\{e_{r}\}$ by syndromes yields all matching pairs with
complexity of order $L\log_{2}L$. Thus, we find a code vector of weight
$d=\delta n$ in any linear $q$-ary code with complexity of order $q^{Fn}$,
where
\begin{equation}
F=h_{q}(\delta)/2.\label{mb-1}%
\end{equation}
For $q$-ary codes on the GV bound, $F^{\ast}=(1-R)/2$.  For stabilizer codes,
the arguments used to prove Corollary \ref{cor:SW} then give exponents
(\ref{4}) and (\ref{5}).
\end{proof}

Note that the MB-technique works for any linear code, unlike other known
techniques provably valid for random codes. For very high rates $R\rightarrow
1,$ this technique yields the lowest complexity exponent known for classical
and quantum codes.

\subsection{Punctured bipartition (PB) technique}

\begin{proposition}
\label{prop:PB} The code distance $\delta n$ of a random quantum stabilizer
code $Q[[n,Rn]]$ can be found with complexity $2^{nF_{\mathrm{PB}}},$ where%
\begin{equation}
F_{\mathrm{PB}}=\textstyle{\frac{2(1+R)}{3+R}}h_{4}(\delta)
\label{eq:punctured-complexity}%
\end{equation}
For random stabilizer codes that meet the quantum GV bound (\ref{eq:GV}),
\begin{equation}
F_{\mathrm{PB}}^{\ast}=(1-R^{2})/(3+R) \label{eq:punctured-GV}%
\end{equation}

\end{proposition}

\begin{proof}
We combine the SW and MB techniques, similarly to the soft-decision decoding
of \cite{Dumer-2001}. Let $s=\left\lceil 2nR/(1+R)\right\rceil >k.$ Then for
most random $[n,k]$ codes $C,$ all $n$ punctured codes $C_{I(i,s)}$ are linear
random $[s,k]$-codes. Also, any codeword of weight $d$ has average weight
$v=\left\lfloor ds/n\right\rfloor $ in some window $I(i,s)$. For simplicity,
let $s$ and $v$ be even. We then apply the MB technique and consider all
vectors $e_{\ell}$ and $e_{r}$ of weight $v/2$ on each window $I(i,s)$. The
corresponding sets have the size
\[
L_{s}=(q-1)^{v/2}\left(  _{v/2}^{s/2}\right)  .
\]
We then select all matching pairs $(e_{\ell},e_{r})$ with the same syndrome.
The result is the list $\{e\}$ of code vectors of weight $v$ in the punctured
$[s,k]$ code $C_{I(i,s)}$. For a random $[s,k]$ code, this list $\{e\}$ has
the expected size of order
\[
L_{v}=(q-1)^{v}\left(  _{v}^{s}\right)  /q^{s-k}%
\]
Each vector of the list $\{e\}$ is re-encoded to the full length $n$. For each
$d=1,2,\ldots$, we stop the procedure once we find a re-encoded vector of
weight $d$. The overall complexity has the order of $L_{v}+L_{s}$. It is easy
to verify \cite{Dumer-2001} that for codes that meet the GV bound, our choice
of parameter $s$ gives the same order $L_{v}\sim$ $L_{s}$ and minimizes the
sum $L_{v}+L_{s}$ to the order of $q^{F_{\ast}n}$, where
\begin{equation}
F^{\ast}=h_{q}(\delta)R/(1+R)=R(1-R)/(1+R).\label{eq-punctured}%
\end{equation}
To proceed with quantum codes $Q[[n,Rn]],$ observe that our parameter $s$
again depends on the effective code rate $R^{\prime}=(1+R)/2$. For stabilizer
codes, this change yields exponent (\ref{eq:punctured-complexity}), which
gives (\ref{eq:punctured-GV}) if codes meet the quantum GV bound.\smallskip
\end{proof}

For codes of rate $R\rightarrow1$ that meet the GV bound, the PB technique
gives the lowest known exponents $F_{\mathrm{PB}}^{\ast}$ (for stabilizer
codes) and $F^{\ast}$ (for classical $q$-ary codes). However, no complexity
estimates have been proven for specific code families.

Finally, consider the narrower Calderbank-Shor-Steane (CSS) class of quantum
codes. Here a parity check matrix is a direct sum $H=G_{x}\oplus\omega G_{z}$,
and the commutativity condition simplifies to $G_{x}G_{z}^{T}=0$. A CSS code
with $\mathrm{rank}\mathop G_{x}=\mathrm{rank}\mathop G_{z}=(n-k)/2$ has the
same effective rate $R^{\prime}=(1+R)/2$ since both codes include $k^{\prime
}=n-(n-k)/2=(n+k)/2$ information bits. Since CSS codes are based on binary
codes, their complexity exponents $F(R,\delta)$ can be obtained from
(\ref{sw0}), (\ref{mb-1}), and (\ref{eq-punctured}) with parameters $q=2$ and
$R^{\prime}=(1+R)/2$. Here we can also use the GV bound, which reads for CSS
codes as follows \cite{Calderbank-Shor-1996}
\begin{equation}
R=1-2h_{2}(\delta). \label{eq:GV-CSS}%
\end{equation}

\subsection{Covering set (CS) technique\label{sec:CS}}

This probabilistic technique was proposed in \cite{Prange-1962} and has become
a benchmark in code-based cryptography since classical paper
\cite{McEliece-1978}. This technique  lowers all three complexity estimates
(\ref{1}), (\ref{4}), and (\ref{eq:punctured-complexity}) except for code
rates $R\rightarrow1$. The CS technique has also been studied for distance
verification of specific code families (see \cite{Lee-1988} and
\cite{Leon-1988}); however, provable results
\cite{Kruk-1989,Coffey-Goodman-1990} are only known for generic random codes.

Let $C[n,k]$ be some $q$-ary random linear code with an $r\times n$ parity
check matrix $H$, $r=n-k$. Consider some subset $J$ of $\rho\leq r$ positions
and the complementary subset $I$ of $g\geq k$ positions. Then the
\emph{shortened} code $C_{J}=\{c_{J}:c_{I}=0\}$ has the parity-check matrix
$H_{J}$ of size $r\times\rho.$ We say that matrix $H_{J}$ has \emph{co-rank}
$b\left(  H_{J}\right)  =\rho-\mathop\mathrm{rank}H_{J}.$ Note that $b\left(
H_{J}\right)  =\dim C_{J},$ which is the dimension of code $C_{J}.$

\begin{proposition}
\label{prop:CS} The code distance $\delta n$ of a random quantum stabilizer
code $Q[[n,Rn]]$ can be found with complexity $2^{nF_{\mathrm{CS}}},$ where%
\begin{equation}
F_{\mathrm{CS}}=h_{2}(\delta)-\textstyle\left(  \frac{1-R}{2}\right)
h_{2}\left(  \frac{2\delta}{1-R}\right)  \label{3}%
\end{equation}

\end{proposition}

\begin{proof}
First, consider a $q$-ary code $C[n,k]$. We randomly choose the sets $J$ of
$r$ positions to cover every possible set of $d<r$ non-zero positions. To do
so, we need no less than
\[
T(n,r,d)={\binom{n}{d}}/{\binom{r}{d}}%
\]
sets $J$. On the other hand, the proof of Theorem 13.4 of \cite{Erdos-book}
shows that a collection of
\begin{equation}
T=T(n,r,d)n\ln n \label{eq:112}%
\end{equation}
random sets $J$ fails to yield such an $(n,r,d)$-covering with a probability
less than $e^{-n\ln n}$. It is also well known that most $r\times n$ matrices
$H$ (excluding a fraction ${\binom{n}{r}}^{-1}\,$ of them) yield small
co-ranks
\begin{equation}
0\leq b_{J}\leq b_{\max}=\textstyle\sqrt{2\log_{q}{\binom{n}{r}}} \label{11}%
\end{equation}
for all square submatrices $H_{J}$, $|J|=r.$

Given an $(n,r,d)$-covering $W$, the CS procedure inspects each set $J\in W$
and discards code $C$ if $\dim C_{J}>b_{\max}.$ Otherwise, it finds the
lightest codewords on each set $J$. To do so, we first perform Gaussian
elimination on $H_{J}$ and obtain a new $r\times r$ matrix $\mathcal{H}_{J}$
that has the same co-rank $b\left(  H_{J}\right)  .$ Let $\mathcal{H}_{J}$
include $r-b_{J}$ unit columns $u_{i}=(0\ldots01_{i}0\ldots0)$ and $b_{J}$
other (linearly dependent) columns $g_{j}$. All $r$ columns have zeroes in the
last $b_{J}$ positions. If $b_{J}=0$ in trial $J,$ then $C_{J}=0$ and we
proceed further. If $b_{J}>0,$ the CS algorithm inspects $q^{b_{J}}-1$ linear
combinations (LC) of columns $g_{j}$. Let $LC(p)$ denote some LC that includes
$p$ columns $g_{j}$. If this $LC(p)$ has weight $w,$ we can nullify it by
adding $w$ unit columns $u_{i}$ and obtain a codeword $c$ of weight $w+p$. The
algorithm ends once we find a codeword of weight $w+p=d$, beginning with $d=2$.

For codes that satisfy condition (\ref{11}), the CS algorithm has the
complexity order of $n^{3}q^{b_{\max}}T(n,r,d)$ that is defined by $T(n,r,d)$.
For any $q,$ this gives complexity $2^{nF}$ with exponent%
\begin{equation}
F=(1-R)\bigl[1-h_{2}\left(  \delta/(1-R)\right)  \bigr]\label{gener}%
\end{equation}
For a stabilizer code $[[n,Rn]]$, we obtain (\ref{3}) using the quaternary
code $C^{\bot}$ with the effective code rate $R^{\prime}=(1+R)/2$.
\end{proof}


For stabilizer codes that meet the quantum GV bound (\ref{eq:GV}), exponent
$F_{\mathrm{CS}}$ of (\ref{3}) reaches its maximum $F_{\mathrm{max}}%
\approx0.22$ at $R=0$. Their binary counterparts yield exponent (\ref{gener})
that achieves its maximum $0.119$ at $R\approx1/2$.\smallskip

\emph{Discussion. } Fig.~\ref{fig:cmp} exhibits different complexity exponents
computed for stabilizer codes that meet the quantum GV bound. The CS technique
gives the best performance for most code rates $R<1$, while the two
bipartition techniques perform better for high code rates $R,$ which are close
to 1. Indeed, equations~(\ref{5}) and (\ref{eq:punctured-GV}) scale linearly
with $1-R$, unlike the CS technique that yields a logarithmic slope, according
to (\ref{3}). \begin{figure}[ptbh]
\centering
\includegraphics[width=0.6\columnwidth]{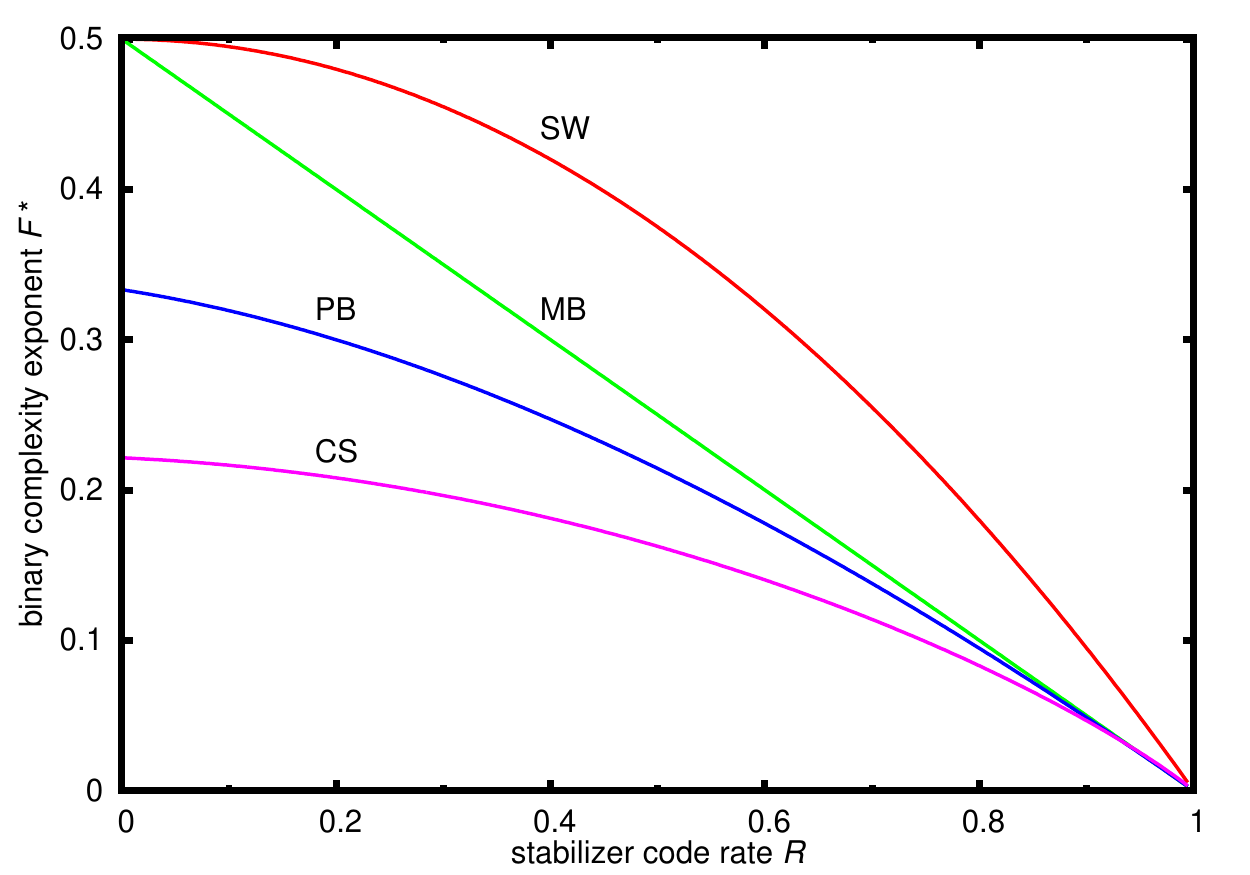}\caption{Complexity exponents of
the four generic decoding techniques applied to quantum codes that meet the
quantum GV bound (\ref{eq:GV}). SW: sliding window, (\ref{eq:SW-GV-complexity}%
), MB: matching bipartition, (\ref{5}), PB: punctured bipartition,
(\ref{eq:punctured-GV}), and CS: covering set, (\ref{3}).}%
\label{fig:cmp}%
\end{figure}

More generally, the above algorithms correct the received vector $y$ into the
list of codewords located at distance $d$ from $y$. In this regard, they are
similar to the maximum likelihood decoding of vector $y$ within a given
distance. For example, given an error syndrome $h\neq0,$ MB\ technique still
forms the sets of vectors $\{e_{\ell}\}$ and $\{e_{r}\}.$ It also derives the
syndromes $h(e_{\ell})$, but uses the syndromes $h(e_{r})+h$ on the right
half. Similarly, some SW trials will correctly identify errors on the
information blocks and then perform error-free re-encoding. For the CS
algorithm, we also make a slight adjustment and inspect all combinations
$LC(p)+h$. Each combination $LC(p)+h$ of weight $w$ gives an error of weight
$p+w$. It is also important that every trial of the CS algorithm needs only
the syndrome $h$ instead of \ the received vector $y.$ Thus, this algorithm
can perform syndrome-based decoding of quantum stabilizer codes.

\section{Distance verification for LDPC codes\label{sec:LDPC}}

Below, we consider two ensembles of binary $(\ell,m)$-LDPC codes with
$m\geq\ell\geq3$. Codes in these ensembles are defined by the binary
equiprobable $r\times n$ parity-check matrices $H$. In ensemble $\mathbb{A}%
(\ell,m),$ matrices $H$ have all columns of weight $\ell$ and all rows of
weight $m=\ell n/r$. This ensemble also includes a smaller LDPC ensemble
$\mathbb{B}(\ell,m)$ originally proposed by Gallager \cite{Gallager}. For each
code in $\mathbb{B}(\ell,m)$, its parity-check matrix $H$ is divided into
$\ell$ horizontal blocks $H_{1},\ldots,H_{\ell}$ of size $\frac{r}{\ell}\times
n$. Here the first block $H_{1}$ consists of $m$ unit matrices of size
$\frac{r}{\ell}\times\frac{r}{\ell}$. Any other block $H_{i}$ is obtained by
some random permutation $\pi_{i}(n)$ of $n$ columns of $H_{1}$. Below, we use
an equivalent description, where block $H_{1}$ also undergoes a random
permutation $\pi_{1}(n)$. Ensembles $\mathbb{A}(\ell,m)$ and $\mathbb{B}%
(\ell,m)$ have similar spectra and achieve the best asymptotic distance for a
given code rate $1-\ell/m$ among the LDPC ensembles studied to date
\cite{Litsyn-2002}.

For brevity, we say below that a linear code $C$ with $N$ non-zero codewords
has \textit{null-free }\emph{size} $N$. We also say that code ensemble
$\mathbb{C}$ is \emph{permutation-invariant} (PI) if any permutation of
positions $\pi$ in any code $C\in\mathbb{C}$ again gives a code $\pi
(C)\in\mathbb{C}.$ In particular, LDPC ensembles are in this class. For any
subset of positions $J$ of size $\mathcal{\rho}=\theta n,$ consider all
shortened codes $C_{J}\in\mathbb{C}_{J}\mathbf{.}$ Then for any PI ensemble
$\mathbb{C},$ all shortened ensembles $\mathbb{C}_{J}$ have the same expected
null-free\textit{ }size $N_{\theta}$ given any $J$ of size $\theta n.$ By
Markov's inequality, for any parameter $t>0$, at most a fraction $\frac{1}{t}$
of the shortened codes $C_{J}$ have null-free\textit{ }size exceeding
$tN_{\theta}$ on any subset $J$.

Note that for LDPC codes, parity checks form non-generic sparse matrices
$H_{J}$. Therefore, below we change the approach of Section \ref{sec:generic}.
In essence, we will relate the size $2^{b_{J}}$ of codes $C_{J}$ to the
erasure-correcting capacity of LDPC codes. In doing so, we extensively use
average weight spectra derived for ensemble $\mathbb{B}(\ell,m)$ in
\cite{Gallager} and for ensemble $\mathbb{A}(\ell,m)$ in \cite{Litsyn-2002}.
This analysis can readily be extended to other ensembles with  known average
weight spectra. The following results are well known and will be extensively
used in our complexity estimates.

Let $\alpha=\ell/m=1-R$. For any parameter $\beta\in\lbrack0,1],$ the
equation
\begin{equation}
\frac{(1+t)^{m-1}+(1-t)^{m-1}}{(1+t)^{m}+(1-t)^{m}}=1-\beta\label{t1}%
\end{equation}
has a single positive root $t$ as a function of $\beta$. Below we use the
parameter
\begin{equation}
q({\alpha,}\beta)=\alpha\log_{2}\frac{(1+t)^{m}+(1-t)^{m}}{2t^{\beta m}%
}-\alpha mh_{2}(\beta), \label{p1}%
\end{equation}
where we also take $q({\alpha,}\beta)=-\infty$ if $m$ is odd and $\beta
\geq1-m^{-1}$. Then Theorem 4 of \cite{Litsyn-2002} shows that a given
codeword of weight $\beta n$ belongs to some code in $\mathbb{A}(\ell,m)$ with
probability $P({\alpha,}\beta)$ such that
\begin{equation}
\underset{n\rightarrow\infty}{\lim}\textstyle\frac{1}{n}\log_{2}P({\alpha
,}\beta)=q({\alpha,}\beta) \label{p2}%
\end{equation}

\begin{lemma}
\label{lm:ldpc}For any given subset $J$ of size $\theta n$, where $\theta
\leq1,$ codes $C_{J}(\ell,m)$ of the shortened LDPC ensembles $\mathbb{A}%
(\ell,m)$ or $\mathbb{B}(\ell,m)$ have the average null-free\textit{ }size
$N_{\theta}$ such that
\begin{equation}
\underset{n\rightarrow\infty}{\lim}\textstyle\frac{1}{n}\log_{2}N_{\theta
}=f(\theta) \label{punct}%
\end{equation}
where%
\begin{equation}
f(\theta)=\max_{0<\beta<1}\left\{  q({\alpha},{\beta\theta})+{\theta
h_{2}(\beta)}\right\}  \label{f}%
\end{equation}

\end{lemma}

\begin{proof}
For any set $J$ of size $\theta n,$ consider codewords of weight $\beta{\theta
n}$ that have support contained on $J$. For any $\beta\in(0,1]$, codes in
$\mathbb{A}_{J}(\ell,m)$ contain the average number
\begin{equation}
N_{\theta}\left(  \beta\right)  =P({\alpha},{\beta\theta}){\binom{\theta
n}{\beta\theta n}} \label{n2}%
\end{equation}
of such codewords of weight $\beta\theta n$. Then
\begin{equation}
\frac{1}{n}\log_{2}N_{\theta}\,{\sim}\,\frac{1}{n}\max_{\beta<1}\log
_{2}\,N_{\theta}\left(  \beta\right)  \,{\sim}\,\max_{\beta<1}\left\{
q({\alpha},{\beta\theta})+{\theta h_{2}(\beta)}\right\}  \smallskip
\smallskip\label{p3}%
\end{equation}
which gives asymptotic equalities (\ref{punct}) and (\ref{f}).
\end{proof}

In the sequel, we show that verification complexity is defined by two
important parameters, $\delta_{\ast}$ and $\theta_{\ast},$ which are the roots
of the equations%
\begin{equation}%
\begin{tabular}
[c]{l}%
$\delta_{\ast}:h_{2}(\delta_{\ast})+q({\alpha},{\delta_{\ast}})=0\smallskip
\smallskip$\\
$\theta_{\ast}:f(\theta)=0.\,$%
\end{tabular}
\ \ \ \ \ \ \ \label{dist1}%
\end{equation}
\emph{Discussion. } Note that $\delta_{\ast}$ is the average relative code
distance in  ensemble $\mathbb{A}(\ell,m)$. Indeed, for ${\theta=1,}$ equality
(\ref{n2}) shows that the average number of codewords $N_{\theta}(\beta)$ of
length $n$ and weight $\beta n$ has the asymptotic order
\begin{equation}
\textstyle\frac{1}{n}\log_{2}N(\beta)\,{\sim}\,h_{2}(\beta)+q({\alpha,}%
\beta)\label{spec}%
\end{equation}
Parameter $\theta_{\ast}$ bounds from below the erasure-correcting capacity of
LDPC codes. Indeed, $f(\theta)<0$ in (\ref{f}) and $N_{\theta}=2^{nf(\theta
)}\rightarrow0$ for any $\theta<\theta_{\ast}.$ Thus, most codes
$C\in\mathbb{A}(\ell,m)$ yield only the single-vector codes $C_{J}%
(\ell,m)\equiv0$ and correct any erased set $J$ of size $\theta n.$ The upper
bounds on the erasure-correcting capacity of LDPC codes are also very close to
$\theta_{\ast}$ and we refer to  papers \cite{Rich-Urb-2001,Pishro-2004},
where this capacity is discussed in detail.

More generally, consider any PI ensemble $\mathbb{C}$ of $q$-ary linear codes.
We say that $\theta_{\ast}$ is the \textit{erasure-correcting capacity }for
ensemble\textit{ }$\mathbb{C}$\textit{ }if for any $\epsilon>0$ the shortened
subcodes $C_{J}$ of length $\theta n,$ $n\rightarrow\infty,$ have expected
size $N_{\theta}$ such that
\begin{equation}
\left\{
\begin{array}
[c]{ll}%
N_{\theta}\rightarrow0, & \text{if}\;\theta\leq\theta_{\ast}-\epsilon
\smallskip\\
N_{\theta}\geq1, & \text{if}\;\theta\geq\theta_{\ast}+\epsilon
\end{array}
\right.  \label{capac}%
\end{equation}
Without ambiguity, we will use the same notation $\theta_{\ast}$ for any lower
bound on the erasure-correcting capacity (\ref{capac}). In this case, we still
have asymptotic condition $N_{\theta}\rightarrow0$ for any $\theta\leq
\theta_{\ast}-\epsilon,$ which is the only condition required for our further
estimates. In particular, we use parameter $\theta_{\ast}$ of (\ref{dist1})
for the LDPC ensembles $\mathbb{A}(\ell,m)$ or $\mathbb{B}(\ell,m).$

For any code rate $R=1-\ell/m$, $\delta_{\ast}$ of (\ref{dist1}) falls below
the GV distance $\delta_{GV}(R)$ of random codes (see \cite{Gallager} and
\cite{Litsyn-2002}). For example, $\delta_{\ast}\sim0.02$ for the
$\mathbb{A}(3,6)$ LDPC ensemble of rate $R=1/2$, whereas $\delta_{GV}\sim
0.11$. On the other hand, $\theta_{\ast}$ also falls below the
erasure-correcting capacity $1-R$ of random linear codes. For example,
$\theta_{\ast}=0.483$ for the ensemble $\mathbb{A}(3,6)$ of LDPC codes of rate
0.5. In our comparison of {LDPC codes and random linear codes, } we will show
that the smaller distances {$\delta_{\ast}$ reduce the verification complexity
for LDPC codes, }despite their weaker erasure-correcting capability
$\theta_{\ast}$ for any code rate $R.$

\subsection{Deterministic techniques for the LDPC ensembles.
\label{sec:SW-LDPC}}

\begin{proposition}
\label{prop:SW-1}Consider any PI ensemble of codes $\mathbb{C}$ with the
average relative distance {$\delta_{\ast}$} and the erasure-correcting
capacity $\theta_{\ast}$. For most codes $C\in\mathbb{C}$, the SW technique
performs distance verification with complexity of exponential order $q^{Fn}$
or less, where
\begin{equation}
F=(1-\theta_{\ast})h_{q}(\delta_{\ast}) \label{sw-ldpc1}%
\end{equation}

\end{proposition}

\begin{proof}
We use the generic SW technique but select sliding windows $I=I(i,s)$ of
length $s=(1-\theta_{\ast}+\varepsilon)n$. Here $\varepsilon>0$ is a parameter
such that $\varepsilon\rightarrow0$ as $n\rightarrow\infty$. For a given
weight $d=\delta_{\ast}n,$ we again inspect each window $I(i,s)$ and take all
$L$ punctured vectors $c_{I(i,s)}$ of average weight $v=\left\lfloor
\delta_{\ast}s\right\rfloor .$ Thus,$^{{}}$%
\[
\textstyle\frac{1}{n}\log_{q}L{\,{\sim}}\textstyle\frac{1}{n}\log
_{q}\textstyle(q-1)^{v}{\binom{s}{v}}\sim(1-\theta_{\ast}+\varepsilon
)h_{q}(\delta_{\ast})
\]
For each vector $c_{I(i,s)},$ we recover symbols on the complementary set
$J=\overline{I}$ of size $(\theta_{\ast}-\varepsilon)n,$ by correcting
erasures in a given code $C\in\mathbb{C}$. This recovery is done by encoding
each vector $c_{I(i,s)}$ into $C$ and gives the codeword list of expected size
$N_{\theta}.$ Thus, codes $C$ have the average complexity of $n^{3}N_{\theta
}L$ combined for all $n$ subsets $I.$ Then only a fraction $n^{-1}$ of such
codes may have a complexity above $n^{4}N_{\theta}L$. This gives
(\ref{sw-ldpc1}) as $\varepsilon\rightarrow0$.\smallskip\
\end{proof}

We proceed with the MB technique, which can be applied to any linear code. For
$q$-ary codes, the MB technique gives the complexity exponent $F=${$h_{q}$%
}${(\delta_{\ast})}/2.$ Combining Propositions \ref{prop:MB} and
\ref{prop:SW-1}, we have

\begin{corollary}
Distance verification for most LDPC codes in the ensembles $\mathbb{A}%
(\ell,m)$ or $\mathbb{B}(\ell,m)$ can be performed with the complexity exponent
\begin{equation}
F=\min\{{(1-\theta}_{\ast}{)}h_{2}{(\delta_{\ast})},h_{2}{(\delta_{\ast})}/2\}
\label{comb}%
\end{equation}
where parameters {$\delta_{\ast}$} and ${\theta}_{\ast}$ are defined in
(\ref{dist1}).\smallskip\
\end{corollary}

The PB technique can also be applied to LDPC codes without changes. However,
its analysis becomes more involved. Indeed, syndrome-matching in the PB
technique yields some punctured $(s,k)$ codes $C_{I(i,s)}$, which are no
longer LDPC codes. However, we can still use their weight spectra, which are
defined by the original ensemble $\mathbb{C}$ and were derived in
\cite{Hsu-2008}. Here we omit lengthy calculations and proceed with a more
efficient CS technique.

\subsection{CS technique for LDPC ensembles \label{sec:atyp-LDPC}}

Below we estimate the complexity of the CS technique for any LDPC code
ensemble. Recall from Section \ref{sec:CS} that for most linear random $[n,k]$
codes, all shortened codes $C_{J}$ of length $n-k$ have non-exponential size
$2^{b_{J}}$. This is not proven for the LDPC codes or any other ensemble of
codes. Therefore, we modify the CS technique to extend it to these non-generic
ensembles. In essence, we leave aside the specific structure of parity-check
matrices $H_{J}$. Instead, we use the fact that atypical codes $C_{J}$ with
large size $2^{b_{J}}$ still form a very small fraction of all codes $C_{J}$.

\begin{proposition}
\label{prop:CS-1}Consider any PI ensemble $\mathbb{C}$ of $q$-ary linear codes
with the average relative distance {$\delta_{\ast}$} and the
erasure-correcting capacity $\theta_{\ast}$. For most codes $C\in\mathbb{C}$,
the CS technique performs distance verification with complexity of exponential
order $2^{Fn}$ or less, where
\begin{equation}
F=h_{2}(\delta_{\ast})-{\theta}_{\ast}h_{2}(\delta_{\ast}/{\theta}_{\ast})
\label{cs-exp}%
\end{equation}

\end{proposition}

\begin{proof}
We now select sets $J$ of $s={\theta n}$ positions, where ${\theta
=\theta_{\ast}-\varepsilon}$ and ${\varepsilon\rightarrow0}$ as $n\rightarrow
\infty$. To find a codeword of weight $d$ in a given code $C\in\mathbb{C},$ we
randomly pick up $T=(n\ln n){\binom{n}{d}}/{\binom{s}{d}}$ sets $J.$ For any
$J,$ the shortened code ensemble $\mathbb{C}_{J}$ has the expected null-free
size $N_{\theta}\rightarrow0$. Let $\mathbb{C}_{J}(b)\subset\mathbb{C}_{J}$ be
a sub-ensemble of codes $C_{J}(b)$ that have null-free size $q^{b}-1$ for some
$b=0,\ldots,{\theta n.}$ Also, let $\alpha_{\theta}(b)$ be the fraction of
codes $C_{J}(b)$ in the ensemble $\mathbb{C}_{J}$. Then
\begin{equation}
N_{\theta}=\sum_{b=0}^{{\theta n}}\left(  q^{b}-1\right)  \alpha_{\theta
}(b)\label{ave1}%
\end{equation}
For each code $C_{J}(b),$ we again apply Gaussian elimination to its
parity-check matrix $H_{J}$ of size $r\times s$. Similarly to the proof of
Proposition\textit{ }\ref{prop:CS}, we obtain the diagonalized matrix
$\mathcal{H}_{J}$, which consists of $s-b$ unit columns $u_{i}=(0\ldots
01_{i}0\ldots0)$ and $b$ other columns $g_{j}$. To find the lightest codewords
on a given set $J$, we again consider all $q^{b}-1$ non-zero linear
combinations of $b$ columns $g_{j}$. For any given code $C_{J}(b)$, this gives
complexity of order $\mathcal{D}_{\theta}(i)\leq n^{3}+rb(q^{b}-1)\leq
n^{3}(q^{b}-1)$. Taking all codes $C_{J}(b)$ for $b=0,\ldots,{\theta n}$ on a
given set $J$, we obtain the expected complexity
\begin{equation}
\mathcal{D}_{\theta}=\sum_{b=0}^{{\theta n}}n^{3}(q^{b}-1)\alpha_{\theta
}(b)=n^{3}N_{\theta}\label{ave2}%
\end{equation}
Thus, the CS algorithm has the expected complexity $\mathcal{D}_{ave}%
=n^{3}TN_{\theta}$ for all $T$ sets $J$. Then only a vanishing fraction
$N_{\theta}/n$ of codes $C$ have complexity $\mathcal{D}\geq n^{4}T,$ which
gives the exponent $F\leq\lim\frac{1}{n}\log_{2}\left(  n^{4}T\right)  $ of
(\ref{cs-exp}) for most codes.\ \smallskip
\end{proof}

\emph{Discussion. } Note that Propositions \ref{prop:SW-1} and \ref{prop:CS-1}
employ PI code ensembles $\mathbb{C}$. This allows us to consider all sets $J$
of ${\theta n}$ positions and output \emph{all } codewords of weight $d$ for
most codes $C\in\mathbb{C}.$ If we replace this adversarial model with a less
restrictive channel-coding model, we may correct \emph{most} errors of weight
$d$ instead of all of them. Then we also remove the above restrictions on
ensembles $\mathbb{C}$. Indeed, let us re-define $N_{\theta}$ as the null-free
size of codes $C_{J}$ averaged over all codes $C\in\mathbb{C}$ and all subsets
$J$ of size $\theta n$. Then we use the following statement:

\begin{lemma}
\label{lm:erasure}Let ensemble $\mathbb{C}$ have vanishing null-free size
$N_{\theta}\rightarrow0$ in the shortened codes $C_{J}$ of length $\theta n$
as $n\rightarrow\infty.$ Then most codes $C\in\mathbb{C}$ \ correct most
erasure subsets $J,$ with the exception of vanishing fraction $\sqrt
{N_{\theta}}$ of \ codes $C$ and subsets $J.$
\end{lemma}

\begin{proof}
A code $C\in\mathbb{C}$ fails to correct some erasure set $J$ \ of weight
$\theta n$ if and only if code $C_{J}$ has $N_{J}(C)\geq1$ non-zero codewords.
\ Let $M_{\theta}$ be the average fraction of such codes $C_{J}$ taken over
all codes $C$ and all subsets $J.$ Note that $M_{\theta}\leq N_{\theta}$. Per
Markov's inequality, no more than a fraction $\sqrt{M_{\theta}}$ of codes $C$
may leave a fraction $\sqrt{M_{\theta}}$ of sets $J$ \ uncorrected. \smallskip
\end{proof}

Finally, we summarize the complexity estimates for classical binary LDPC codes
in Fig. \ref{fig:ldpc}. For comparison, we also plot two generic exponents
valid for most linear binary codes. The first exponent
\begin{equation}
F=\min\{R(1-R),(1-R)/2\} \label{gener0}%
\end{equation}
combines the SW and MB algorithms, and the second exponent (\ref{gener})
represents the CS algorithm. For LDPC codes, we similarly consider the
exponent (\ref{comb}) that combines the SW and MB algorithms and the exponent
(\ref{cs-exp}) that represents the CS algorithm for the LDPC codes. Here we
consider ensembles $\mathbb{A}(\ell,m)$ or $\mathbb{B}(\ell,m)$ for various
LDPC $(\ell,m)$ codes with code rates ranging from $0.125$ to $0.8.$ With the
exception of low-rate codes, all LDPC codes of Fig. \ref{fig:ldpc} have
substantially lower distances than their generic counterparts. This is the
reason LDPC codes also achieve an exponentially smaller complexity of distance
verification despite their lower erasure-correcting capacity.
\begin{figure}[ptbh]
\centering
\includegraphics[width=0.65\columnwidth]{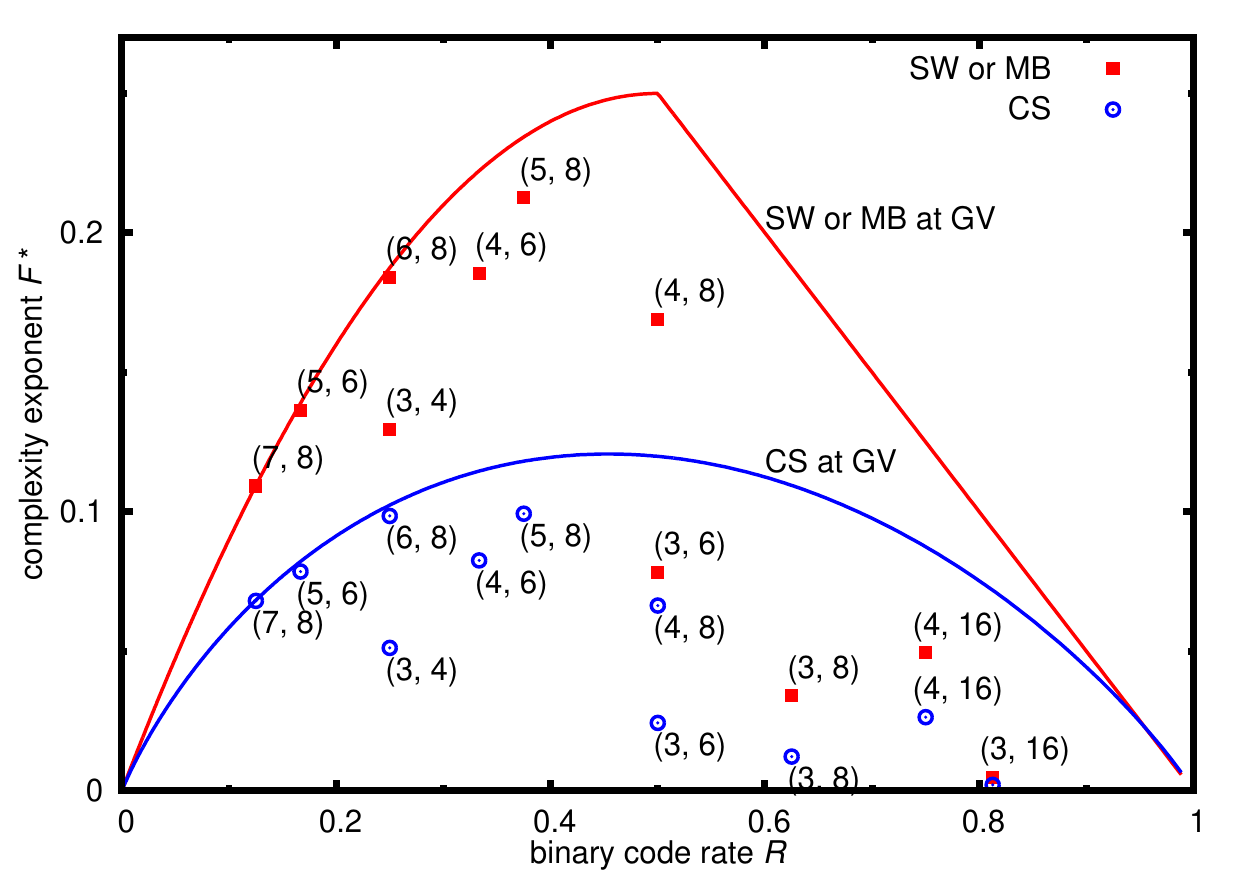} \caption{Complexity
exponents for the binary codes meeting the GV bound and for some ($\ell
,m)$-regular LDPC codes as indicated. \textquotedblleft SW or
MB\textquotedblright\ stands for deterministic techniques from Eq.~(\ref{comb}%
) for LDPC codes, or Eq.~(\ref{gener0}) for codes meeting the GV bound, and CS
stands for covering set technique, Eq.~(\ref{cs-exp}) for LDPC codes, or
Eq.~(\ref{gener}) for codes meeting the GV bound.}%
\label{fig:ldpc}%
\end{figure}

\section{Irreducible-cluster (IC) technique\label{sec:LC-LDPC}}

The complexity estimates of Sec.~\ref{sec:LDPC} rely on the average weight
distributions of binary $(\ell,m)$-regular LDPC codes and hold for most codes
in the corresponding ensembles. Here we suggest a deterministic
distance-verification technique, which is applicable to any $q$-ary $(\ell
,m)$-regular LDPC code, quantum or classical. First, we define irreducible
codewords.\smallskip

\begin{definition}
\label{def:irreducible-cw} Given a linear $q$-ary code $\mathcal{C}_{q}$, we
say that a codeword $c$ is irreducible if it cannot be represented as a linear
combination of two codewords with non-overlapping supports.
\end{definition}

Our technique is based on the following simple lemma.

\begin{lemma}
\label{th:basic}\cite{Dumer-Kovalev-Pryadko-bnd-2015} A minimum-weight
codeword of a linear code $\mathcal{C}_{q}$ is irreducible.
\end{lemma}

\emph{IC algorithm: general description.} Let a $q$-ary $(\ell,m)$-regular
LDPC code be defined by a list $\mathcal{L}$ of parity checks $b$ with
supports $J_{b}$ of size $m.$ The following algorithm finds an irreducible
codeword $c$ of weight $d$. The algorithm performs multiple runs and includes
a variable number $\omega\leq d-1$ of steps in each run. The initial step
$i=0$ of each run is given a position $j_{0}=0,\ldots,n-1$ and the symbol
$c_{j_{0}}=1$. The input to each consecutive step $i$ includes some previously
derived sub-vector $c(J_{i})$ with its support $J_{i}.$ It also includes the
ordered sublist $\mathcal{N}_{i}\mathcal{\subset L}$ of all parity checks $b$
unsatisfied by sub-vector $c(J_{i}).$ Then step $i$ extends vector $c(J_{i})$
with some non-overlapping subset $c(I_{i})$ of $v_{i}$ new non-zero symbols.
The extension $I_{i},c(I_{i})$ is chosen to make the first parity check
$b^{(1)}\in\mathcal{N}_{i}$ satisfied on the extended support $J_{i+1}%
=J_{i}\cup I_{i}:$
\begin{equation}
\sum_{j\in J_{^{i}}}b_{j}^{(1)}c_{j}+\sum_{j\in I_{^{i}}}b_{j}^{(1)}c_{j}=0
\label{rest0}%
\end{equation}
The result is the extended vector $c(J_{i+1})$ and the new list $\mathcal{N}%
_{i+1}$ of parity checks unsatisfied by $c(J_{i+1}).$ Clearly, $\mathcal{N}%
_{i+1}$ excludes parity check $b^{(1)}.$ It may also drop some other checks in
$\mathcal{N}_{i}$, which were satisfied in step $i,$ but may include new
parity checks, which become unsatisfied due to the newly added symbols. Note
that a parity check dropped in step $i$ may later re-appear in some list
$\mathcal{N}_{s},$ $s>i+1.$ Each run must satisfy restrictions (\ref{rest0})
for all steps and end with $d$ symbols, thus
\begin{equation}
\textstyle\sum_{i=1}^{\omega}v_{i}=d-1 \label{rest1}%
\end{equation}
Each run ends with a complete selection list $\left\{  I_{i},c(I_{i}%
)\,|\,i=0,...,\omega\right\}  $ and gives a codeword of weight $d$ if the list
$\mathcal{N}_{\omega+1}$ is empty. For a quantum stabilizer code, we also
verify the restriction $c\in C^{\bot}\smallsetminus C.$ Given no codeword of
weight $d$, we proceed with a new run, which employs a new selection list$.$
We will now limit possible choices of all vectors $c(I_{i}).$

\emph{Additively irreducible selection. } We say that a new selection $I,c(I)$
of non-zero symbols is \emph{additively irreducible (AI) }for a parity-check
$b$ if any non-empty subset $I^{\prime}\subset I$ satisfies restriction
\begin{equation}
\textstyle\sum_{j\in I^{\prime}}b_{j}c_{j}\neq0 \label{q-parity}%
\end{equation}
From now on, any selection list $\left\{  I_{i},c(I_{i})\,|\,i=0,...,\omega
\right\}  $ must also satisfy restrictions (\ref{q-parity}) in each step $i.$
We proceed with the following observations.\smallskip

\textit{A.} If an $AI$ vector satisfies parity check $b^{(1)}$, then no
smaller subset $c(I^{\prime})$ can do so. Indeed, let restrictions
(\ref{rest0}) hold on the sets $I$ and $I^{\prime}\subset I.$ Then we obtain
equality $\sum b_{j}c_{j}=0$ on the subset $I\smallsetminus I^{\prime},$ which
contradicts (\ref{q-parity}). We also see that for any reducible vector $c(I)$
that satisfies the current check $b^{(1)}$, there exists its sub-vector
$c(I^{\prime}),$ which also satisfies $b^{(1)}.$ \smallskip\

\textit{B. }We may process parity checks one-by-one. Indeed, irrespective of
the order in which parity checks are processed, the codewords will satisfy all
parity checks after $w$ steps. We may also set $c_{j_{0}}=1$ in a linear code
$C.$ Our brute-force algorithm begins with a correct choice of $j_{0}$ for
some runs and then exhausts all possible irreducible selections. Thus, in each
step, one of the runs begins with a correct subvector $c(J_{i})$ and then adds
some correct $AI$ subvector $c(I_{i})$. \smallskip

C. The algorithm may terminate only at some codeword of weight $d.$ More
generally, the algorithm can return all (non-collinear) irreducible vectors up
to some weight $D$.

D. If some run fails in step $w$, we can return to step $w-1$ and exhaust all
choices of vectors $c(I_{w-1}).$ Similarly, we can return to step $w-2$ and so
on. This back-and-forth version slightly reduces the overall complexity;
however, it will keep its asymptotic order.\smallskip

Let $N_{v}(q,b)$ denote the number of $q$-ary vectors $c(I)$ of length $v$
that satisfy restrictions (\ref{rest0}) and (\ref{q-parity}). Clearly,
\begin{equation}
N_{v}(q,b)\leq(q-1)^{v-1}\label{q-triv}%
\end{equation}
Below, we use notation $N_{v}(q)$ since we will prove that all parity checks
$b$ give the same number $N_{v}(q,b)\equiv N_{v}(q).$ \ Note also that the AI
restriction (\ref{q-parity}) drastically limits the number $N_{v}(q)$ for
small $q.$ For example, a binary parity check $b^{(1)}$ is satisfied in
(\ref{rest0}) only if $v$ is odd; however, any string of $v\geq3$ ones
includes a subset of two ones and contradicts the AI property (\ref{q-parity}%
). Thus, $v=1$ for $q=2$ and $N_{v}(2)=1.$ \smallskip

We now proceed with complexity estimates. First, we employ a trivial upper
bound (\ref{q-triv}). We further reduce this number in Lemma \ref{th:basic-1}.

Let $\delta_{a,b}$ be the Kronecker symbol, $h=d-1$ and $t=m-1$. Recall that
each run is defined by some set $\left\{  I_{i},c(I_{i})\,|\,i=0,...,\omega
\right\}  .$ Given restriction (\ref{rest1}), the number of runs is bounded
from above by the quantities
\begin{equation}
S_{h}(m,q)\equiv\sum_{\omega\geq1}\,\sum_{v_{i}\in\{1,2,\ldots,t\}}%
\!\!\delta_{h,v_{1}+\ldots+v_{\omega}}\prod_{i=1}^{\omega}N_{v_{i}%
}(q)\textstyle\binom{t}{v_{i}}\label{eq:upper-bound}%
\end{equation}
which have the power-moment generating function
\begin{align}
g(z) &  =1+\sum_{h=1}^{\infty}S_{h}(m,q)z^{h}=\sum_{\omega=0}^{\infty}\left[
T(z)\right]  ^{\omega}=[1-T(z)]^{-1},\label{eq:gener1}\\
T(z) &  \equiv\sum_{h=1}^{t}z^{h}N_{h}(q)\textstyle{\binom{t}{h}%
}.\label{eq:poly1}%
\end{align}
We can now derive the coefficients $S_{h}(m,q)$. This can be done by the
Chernoff bound, similarly to the estimates of \cite{Gallager} or by the
combinatorial technique of \cite{Litsyn-2002}. Instead, we use another simple
technique that employs contour integration and gives the exact formulas for
the coefficients $S_{h}(m,q)$ along with their exponential orders. Namely, let
the denominator $1-T(z)$ in (\ref{eq:gener1}) have $s\leq t$ distinct roots $z_{r}$,
$r=0,1,\ldots,s-1$, with ordered magnitudes $\rho=|z_{0}|\leq|z_{1}|\leq
\ldots\leq|z_{s-1}|$. Then coefficients $S_{h}(m,q)$ can be derived by a
contour integration over a circle of radius $\epsilon<\rho$ around the
origin,
\begin{equation}
S_{h}(m,q)={\frac{1}{2\pi i}}\oint{\frac{dz}{z^{d}}}{\frac{1}{1-T(z)}}%
=-\sum_{r=0}^{s-1}\mathop{\rm Res}\left(  \frac{1}{z^{d}[1-T(z)]}%
,z_{r}\right)  \label{eq:2}%
\end{equation}
where $\mathop{\rm Res}(f(z),a)$ is the residue of $f(z)$ at $a$. For large
weights $d$, the exponential order of $S_{h}(m,q)$ is defined by the root
$z_{0}$, which has the smallest magnitude $\rho$. Next, note that $z_{0}%
=\rho>0$ is strictly positive and non-degenerate, since the coefficients of
$T(z)$ are non-negative. In this case,
\begin{equation}
\mathop{\rm Res}\left(  \frac{1}{z^{d}[1-T(z)]},z_{0}\right)  =-\frac{1}%
{z_{0}^{d}T^{\prime}(z_{0})}\label{eq:root}%
\end{equation}
where ${T}^{\prime}(z)$ is the derivative of the polynomial $T(z)$; it is
non-negative at $z=z_{0}$. This gives the exponential bound
\begin{equation}
S_{h}(m,q)\leq c\rho^{-d}+\mathcal{O}(|z_{1}|^{-d})\sim c[\gamma_{m}%
(m-1)]^{d}\label{eq:Sv-bound}%
\end{equation}
with the complexity exponent $\gamma_{m}\equiv1/[(m-1)\rho]$.\smallskip

We now employ upper bound (\ref{q-triv}). \ In this case, equality
(\ref{eq:poly1}) gives the polynomial
\[
\overline{T}(z)={\frac{1}{q-1}}\left\{  \left[  (q-1)z+1\right]
^{t}-1\right\}
\]
which has the roots
\[
z_{r}=(q^{1/t}e^{2\pi ir/t}-1)/(q-1),\;r=0,1,\ldots,t-1
\]
Thus, the asymptotic expansion (\ref{eq:Sv-bound}) yields the constant
\[
c=\frac{1+(q-1)\rho}{qt\,}=\frac{q^{1/(m-1)}}{q(m-1)}%
\]
and the complexity exponent
\begin{equation}
\overline{\gamma}_{m}=\frac{q-1}{(m-1)\left(  q^{1/(m-1)}-1\right)  }%
\leq\overline{\gamma}_{\infty}={\frac{q-1}{\ln q}}\label{eq:gamma-bnd3}%
\end{equation}
As a side remark, note that the larger values $v_{i}>1$ reduce the number of
terms in the product taken in (\ref{eq:upper-bound}); therefore, they
contribute relatively little to the overall sum $S_{h}(m,q).$ \ It is for this
reason that a simple bound (\ref{q-triv}) can yield a reasonably tight
estimate (\ref{eq:gamma-bnd3}). \

Our next step is to reduce the exponent $\overline{\gamma}_{m}$ by limiting
the number $N_{v}(q,b).$ Let $M_{v}(q)$ denote the set of $q$-ary vectors
$c(I)$ of length $v$ that satisfy the restrictions
\begin{equation}
\sum_{I^{\prime}}c_{j}\neq0\text{ for all }I^{\prime}\subseteq I \label{AI-1}%
\end{equation}
Let $A_{v}(q)$ be the size of $M_{v}(q)$ and $v_{\max}$ be the maximum length
of vectors in $M_{v}(q).$

\begin{lemma}
\label{th:basic-1} The number $N_{v}(q,b)$ of $q$-ary vectors $c(I)$ of length
$v,$ which satisfy restrictions (\ref{rest0}) and (\ref{q-parity}) in a Galois
field $F_{q},$ does not depend on a parity check $b$ and is equal to
$A_{v}(q)/(q-1).$ For any $q=2^{u},$ $v_{\max}=u$ and $N_{v}(q)=(q-2)\cdot
\ldots\cdot(q-2^{v-1}).$ For a prime number $q,$ $v_{\max}=q-1.$
\end{lemma}

\begin{proof}
Let two sets of $q$-ary vectors $c(I,b)$ and $c(I,B)$ of length $v$ satisfy
restrictions (\ref{rest0}) and (\ref{q-parity}) for some parity checks $b$ and
$B.$ Then any such vector $c(I,B)$ has its counterpart $c(I,b)$ with symbols
$c_{j}(I,b)=B_{j}c_{j}(I,b)/b_{j}.$ Thus, the two sets have the same size and
$N_{v}(q,b)=N_{v}(q)$. We can also specify AI-restrictions (\ref{q-parity})
using AI-restrictions (\ref{AI-1}) for the parity check $b^{\ast}=(1,...,1)$
and all subsets $I^{\prime}\subset I.$ \ Now let $\lambda\neq0$ be the value
of the first summand in (\ref{rest0}) for some unsatisfied parity check.
Consider a subset of vectors in $M_{v}(q)$ that satisfy restriction $\sum
_{I}c_{j}=-\lambda.$ This subset has the size $A_{v}(q)/(q-1)$ and satisfies
both restrictions (\ref{rest0}) and (\ref{q-parity}) for the parity check
$b^{\ast}.$ Thus, $N_{v}(q)=A_{v}(q)/(q-1).$

Next, consider the Galois field $F_{q}$ for $q=2^{u}.$ Then the sums in the
left-hand side of (\ref{AI-1}) represent all possible linear combinations over
$F_{2}$ \ generated by $v$ or fewer elements of $M_{v}(q).$ Thus, any symbol
$c_{j}(I)$ must differ from the linear combinations of the previous symbols
$c_{1}(I),...,c_{j-1}(I).$ This gives the size $A_{v}(q)=(q-1)(q-2)\cdot
\ldots\cdot(q-2^{v-1})$ and also proves that $v_{\max}=u.$

For any prime number $q,$ any sum of $s$ elements in (\ref{AI-1}) must differ
from the sums of $t<s$ elements on its subsets. Thus, different sums may take
at most $v_{\max}$ non-zero values for $s=1,...,v_{\max}$ and $v_{\max}\leq
q-1.$ Then $v_{\max}=q-1$ is achieved on the vector $c=(1,...,1)$ of length
$q-1$. \smallskip
\end{proof}

Lemma \ref{th:basic-1} shows that the numbers $N_{v}(q)$ and the lengths
$v_{\max}$ differ substantially for different $q.$ Some of these quantities
are listed in Table \ref{tab:Nv} for small $q$. Table \ref{tab:even-odd} gives
some exponents $\gamma_{m}$ obtained for irreducible clusters, along with the
upper bound $\overline{\gamma}_{\infty}$ (valid for all clusters) in the last
row. We summarize our complexity estimates as follows.\smallskip

\begin{proposition}
\label{prop:IC-1}A codeword of weight $\delta n$ in any $q$-ary $(\ell,m)$
LDPC code can be found with complexity $2^{F_{\mathrm{\mathrm{IC}}}n}$, where
\[
F_{\mathrm{\mathrm{IC}}}=\delta\log_{2}(\gamma_{m}(m-1)),
\]
$\gamma_{m}\in(1,\gamma_{\infty})$ grows monotonically with~$m$ and
$\gamma_{\infty}<\overline{\gamma}_{\infty}=\left(  q-1\right)  /\ln
q$.\smallskip
\end{proposition}

\begin{table}[tbh]
\centering
\begin{tabular}
[c]{c|ccccc}%
$v$ & $q=2$ & $q=3$ & $q=4$ & $q=5$ & $q=8_{\strut}$\\\hline
1 & 1 & 1 & 1 & 1 & 1$^{\strut}$\\
2 & 0 & 1 & 2 & 3 & 6\\
3 &  & 0 & 0 & 4 & 24\\
4 &  &  &  & 1 & 0$_{\strut}$\\
&  &  &  &  &
\end{tabular}
\caption{Number of additively-irreducible $q$-ary strings of length $v$ for
$q=p^{m}$. }%
\label{tab:Nv}%
\end{table}

\begin{table}[tbh]
\centering%
\begin{tabular}
[c]{c|ccccc}%
$m$ & $q=2$ & $q=3$ & $q=4$ & $q=5$ & $q=8_{\strut}$\\\hline
$3$ & 1 & 1.20711 & 1.36603 & 1.5 & 1.82288$^{\strut}$\\
$5$ & 1 & 1.29057 & 1.5 & 1.73311 & 2.27727\\
$10$ & 1 & 1.33333 & 1.56719 & 1.85548 & 2.50514\\
$10^{2}$ & 1 & 1.3631 & 1.61351 & 1.94162 & 2.66259\\
$10^{3}$ & 1 & 1.36574 & 1.61759 & 1.94927 & 2.67647\\
$\infty$ & 1 & 1.36603 & 1.61803 & 1.95011 & 2.67799$_{\strut}$\\%
\begin{tabular}
[c]{l}%
Upper bound\\
$(q-1)/\ln q_{\strut}$%
\end{tabular}
& 1.44270\  & 1.82048\  & 2.16404\  & 2.48534\  & 3.36629 $_{\strut}$%
\end{tabular}
\caption{Coefficient $\gamma_{m}$ of the complexity exponent $\delta\log
_{q}(\gamma_{m}(m-1))$ for different $m$ and $q$.}%
\label{tab:even-odd}%
\end{table}

\textit{Remarks.} The algorithm presented here for linear $q$-ary codes
generalizes an algorithm described in \cite{Dumer-Kovalev-Pryadko-bnd-2015}
for binary codes. It can be also applied to a more general class of $q$-ary
$(\ell,m)$-limited LDPC codes, whose parity check matrices have all columns
and rows of Hamming weights no more than $\ell$ and $m,$ respectively. This
algorithm is also valid for $q$-ary CSS codes, and gives the same complexity
exponent. However, for $q$-ary stabilizer codes, the numbers of additively
irreducible clusters (e.g., from Table \ref{tab:Nv}) have to be increased by
an additional factor of $q^{v}$, $N_{v}^{(\mathrm{stab})}(q)=q^{v}N_{v}(q)$.
As a result, the complexity exponents in Table \ref{tab:even-odd} also
increase, $\gamma_{m}^{(\mathrm{stab})}=q\gamma_{m}$. In particular, for qubit
stabilizer codes, $q=2$, we obtain complexity exponent $\gamma_{m}%
^{(\mathrm{qubit})}=2$.

Also, note that for the existing quantum LDPC codes with distance $d$ of order
$\sqrt{n},$ the presented IC algorithm has the lowest proven complexity among
deterministic algorithms. Indeed, exponent $F_{\mathrm{IC}}$ is linear in the
relative distance $\delta$, whereas deterministic techniques of
Sec.~\ref{sec:generic} give the higher exponents $F\rightarrow\delta
\log(1/\delta)$ in this limit. In this regard, exponent $F_{\mathrm{IC}}$
performs similarly to the CS exponent $F_{\mathrm{CS}}$ of generic codes,
which is bounded by $\delta-\delta\log_{2}(1-R)$ and is linear in $\delta$.

\section{Further extensions}

In this paper, we study provable algorithms of distance verification for LDPC
codes. More generally, this approach can be used for any ensemble of codes
with a given relative distance $\delta_{\ast}$ and erasure-correcting capacity
${\theta}_{\ast}$. \

One particular extension is any \ ensemble of irregular LDPC codes with known
parameters $\delta_{\ast}$ and ${\theta}_{\ast}.$ Note that parameter
${\theta}_{\ast}$ has been studied for both ML decoding and message-passing
decoding of irregular codes \cite{Luby-2001,Rich-Urb-2001,Pishro-2004}. For ML
decoding, this parameter can also be derived using the weight spectra obtained
for irregular codes in papers \cite{Di-2001,Litsyn-2003}. Also, these
techniques can be extended to ensembles of $q$-ary LDPC codes. The weight
spectra of some $q$-ary ensembles are derived in \cite{Til-2009,Yang-2011}.

Another direction is to design more advanced algorithms of distance
verification for LDPC codes. Most of such algorithms known to date for linear
$[n,k]$ codes combine the MB and CS techniques. In particular,  algorithm
\cite{Stern-1989} takes a linear $[n,k]$-code and seeks some high-rate punctured
$[k+\mu,k]$-block that has $\varepsilon\ll k$ errors among $k$ information
bits and $\mu$ error-free parity bits. The search is conducted similarly to
the CS technique. Then the MB technique corrects $\varepsilon$ errors in this
high-rate $[k+\mu,k]$-code. A slightly more efficient algorithm
\cite{Dumer-1991} simplifies this procedure and seeks punctured $[k+\mu
,k]$-code that has $\varepsilon\ll k+\mu$ errors spread across information and
parity bits. In this case, the optimal choice of parameters $\varepsilon$ and
$\mu$ reduces the maximum complexity exponent $F(R)$ to 0.1163. Later, this
algorithm was re-established  in \cite{Sendrier-2009, Bernstein-2011}, with
detailed applications for the McEliece cryptosystem. More recently, the
maximum complexity exponent $F(R)$ has been further reduced to 0.1019 using
some robust MB techniques that allow randomly overlapping partitions
\cite{Becker-2012}. An important observation is that both the MB and CS
techniques can be applied to LDPC codes; therefore, our conjecture is that
provable complexity bounds for distance verification also carry over to these
more advanced techniques.






\bibliographystyle{IEEEtran}
\bibliography{IEEEabrv,lpp,qc_all,more_qc,ldpc1,MyBIB}

\begin{thebibliography}{10}
\providecommand{\url}[1]{#1}
\csname url@samestyle\endcsname
\providecommand{\newblock}{\relax}
\providecommand{\bibinfo}[2]{#2}
\providecommand{\BIBentrySTDinterwordspacing}{\spaceskip=0pt\relax}
\providecommand{\BIBentryALTinterwordstretchfactor}{4}
\providecommand{\BIBentryALTinterwordspacing}{\spaceskip=\fontdimen2\font plus
\BIBentryALTinterwordstretchfactor\fontdimen3\font minus
  \fontdimen4\font\relax}
\providecommand{\BIBforeignlanguage}[2]{{%
\expandafter\ifx\csname l@#1\endcsname\relax
\typeout{** WARNING: IEEEtran.bst: No hyphenation pattern has been}%
\typeout{** loaded for the language `#1'. Using the pattern for}%
\typeout{** the default language instead.}%
\else
\language=\csname l@#1\endcsname
\fi
#2}}
\providecommand{\BIBdecl}{\relax}
\BIBdecl

\bibitem{shor-error-correct}
\BIBentryALTinterwordspacing
P.~W. Shor, ``Scheme for reducing decoherence in quantum computer memory,''
  \emph{Phys. Rev. A}, vol.~52, p. R2493, 1995. [Online]. Available:
  \url{http://link.aps.org/abstract/PRA/v52/pR2493}
\BIBentrySTDinterwordspacing

\bibitem{Knill-Laflamme-1997}
\BIBentryALTinterwordspacing
E.~Knill and R.~Laflamme, ``Theory of quantum error-correcting codes,''
  \emph{Phys. Rev. A}, vol.~55, pp. 900--911, 1997. [Online]. Available:
  \url{http://dx.doi.org/10.1103/PhysRevA.55.900}
\BIBentrySTDinterwordspacing

\bibitem{Bennett-1996}
\BIBentryALTinterwordspacing
C.~Bennett, D.~DiVincenzo, J.~Smolin, and W.~Wootters, ``Mixed state
  entanglement and quantum error correction,'' \emph{Phys. Rev. A}, vol.~54, p.
  3824, 1996. [Online]. Available:
  \url{http://dx.doi.org/10.1103/PhysRevA.54.3824}
\BIBentrySTDinterwordspacing

\bibitem{kitaev-anyons}
\BIBentryALTinterwordspacing
A.~Y. Kitaev, ``Fault-tolerant quantum computation by anyons,'' \emph{Ann.
  Phys.}, vol. 303, p.~2, 2003. [Online]. Available:
  \url{http://arxiv.org/abs/quant-ph/9707021}
\BIBentrySTDinterwordspacing

\bibitem{Dennis-Kitaev-Landahl-Preskill-2002}
\BIBentryALTinterwordspacing
E.~Dennis, A.~Kitaev, A.~Landahl, and J.~Preskill, ``Topological quantum
  memory,'' \emph{J. Math. Phys.}, vol.~43, p. 4452, 2002. [Online]. Available:
  \url{http://dx.doi.org/10.1063/1.1499754}
\BIBentrySTDinterwordspacing

\bibitem{Bombin-MartinDelgado-2006}
\BIBentryALTinterwordspacing
H.~Bombin and M.~A. Martin-Delgado, ``Topological quantum distillation,''
  \emph{Phys. Rev. Lett.}, vol.~97, p. 180501, Oct 2006. [Online]. Available:
  \url{http://link.aps.org/doi/10.1103/PhysRevLett.97.180501}
\BIBentrySTDinterwordspacing

\bibitem{Bombin-2007}
------, ``Optimal resources for topological two-dimensional stabilizer codes:
  Comparative study,'' \emph{Phys. Rev. A}, vol.~76, no.~1, p. 012305, Jul
  2007.

\bibitem{Bombin-MartinDelgado-2007}
\BIBentryALTinterwordspacing
------, ``Homological error correction: {C}lassical and quantum codes,''
  \emph{Journal of Mathematical Physics}, vol.~48, no.~5, p. 052105, 2007.
  [Online]. Available:
  \url{http://scitation.aip.org/content/aip/journal/jmp/48/5/10.1063/1.2731356}
\BIBentrySTDinterwordspacing

\bibitem{Raussendorf-Harrington-2007}
\BIBentryALTinterwordspacing
R.~Raussendorf and J.~Harrington, ``Fault-tolerant quantum computation with
  high threshold in two dimensions,'' \emph{Phys. Rev. Lett.}, vol.~98, p.
  190504, 2007. [Online]. Available:
  \url{http://link.aps.org/abstract/PRL/v98/e190504}
\BIBentrySTDinterwordspacing

\bibitem{Bravyi-Poulin-Terhal-2010}
\BIBentryALTinterwordspacing
S.~Bravyi, D.~Poulin, and B.~Terhal, ``Tradeoffs for reliable quantum
  information storage in {2D} systems,'' \emph{Phys. Rev. Lett.}, vol. 104, p.
  050503, Feb 2010. [Online]. Available:
  \url{http://link.aps.org/doi/10.1103/PhysRevLett.104.050503}
\BIBentrySTDinterwordspacing

\bibitem{Postol-2001}
\BIBentryALTinterwordspacing
M.~S. Postol, ``A proposed quantum low density parity check code,'' 2001,
  unpublished. [Online]. Available: \url{http://arxiv.org/abs/quant-ph/0108131}
\BIBentrySTDinterwordspacing

\bibitem{MacKay-Mitchison-McFadden-2004}
\BIBentryALTinterwordspacing
D.~J.~C. MacKay, G.~Mitchison, and P.~L. McFadden, ``Sparse-graph codes for
  quantum error correction,'' \emph{IEEE Trans. Info. Th.}, vol.~59, pp.
  2315--30, 2004. [Online]. Available:
  \url{http://dx.doi.org/10.1109/TIT.2004.834737}
\BIBentrySTDinterwordspacing

\bibitem{Kovalev-Pryadko-FT-2013}
\BIBentryALTinterwordspacing
A.~A. Kovalev and L.~P. Pryadko, ``Fault tolerance of quantum low-density
  parity check codes with sublinear distance scaling,'' \emph{Phys. Rev. A},
  vol.~87, p. 020304(R), Feb 2013. [Online]. Available:
  \url{http://link.aps.org/doi/10.1103/PhysRevA.87.020304}
\BIBentrySTDinterwordspacing

\bibitem{Dumer-Kovalev-Pryadko-bnd-2015}
\BIBentryALTinterwordspacing
I.~Dumer, A.~A. Kovalev, and L.~P. Pryadko, ``Thresholds for correcting errors,
  erasures, and faulty syndrome measurements in degenerate quantum codes,''
  \emph{Phys. Rev. Lett.}, vol. 115, p. 050502, Jul 2015. [Online]. Available:
  \url{http://link.aps.org/doi/10.1103/PhysRevLett.115.050502}
\BIBentrySTDinterwordspacing

\bibitem{Tillich-Zemor-2009}
J.-P. Tillich and G.~Zemor, ``Quantum {LDPC} codes with positive rate and
  minimum distance proportional to {$\sqrt{n}$},'' in \emph{Proc. IEEE Int.
  Symp. Inf. Theory (ISIT)}, June 2009, pp. 799--803.

\bibitem{Zemor-2009}
\BIBentryALTinterwordspacing
G.~Z{\'e}mor, ``On cayley graphs, surface codes, and the limits of homological
  coding for quantum error correction,'' in \emph{Coding and Cryptology}, ser.
  Lecture Notes in Computer Science, Y.~Chee, C.~Li, S.~Ling, H.~Wang, and
  C.~Xing, Eds.\hskip 1em plus 0.5em minus 0.4em\relax Springer Berlin
  Heidelberg, 2009, vol. 5557, pp. 259--273. [Online]. Available:
  \url{http://dx.doi.org/10.1007/978-3-642-01877-0_21}
\BIBentrySTDinterwordspacing

\bibitem{Couvreur-Delfosse-Zemor-2012}
\BIBentryALTinterwordspacing
A.~Couvreur, N.~Delfosse, and G.~Z{\'e}mor, ``A construction of quantum {LDPC}
  codes from {Cayley} graphs,'' \emph{CoRR}, vol. abs/1206.2656, 2012.
  [Online]. Available: \url{http://arxiv.org/abs/1206.2656}
\BIBentrySTDinterwordspacing

\bibitem{Kovalev-Pryadko-2012}
A.~A. Kovalev and L.~P. Pryadko, ``Improved quantum hypergraph-product {LDPC}
  codes,'' in \emph{Proc. IEEE Int. Symp. Inf. Theory (ISIT)}, July 2012, pp.
  348--352.

\bibitem{Andriyanova-Maurice-Tillich-2012}
I.~Andriyanova, D.~Maurice, and J.-P. Tillich, ``New constructions of {CSS}
  codes obtained by moving to higher alphabets,'' 2012, unpublished.

\bibitem{Kovalev-Pryadko-Hyperbicycle-2013}
\BIBentryALTinterwordspacing
A.~A. Kovalev and L.~P. Pryadko, ``Quantum {K}ronecker sum-product low-density
  parity-check codes with finite rate,'' \emph{Phys. Rev. A}, vol.~88, p.
  012311, July 2013. [Online]. Available:
  \url{http://link.aps.org/doi/10.1103/PhysRevA.88.012311}
\BIBentrySTDinterwordspacing

\bibitem{Bravyi-Hastings-2013}
S.~Bravyi and M.~B. Hastings, ``Homological product codes,'' 2013, unpublished.

\bibitem{Guth-Lubotzky-2014}
\BIBentryALTinterwordspacing
L.~Guth and A.~Lubotzky, ``Quantum error correcting codes and 4-dimensional
  arithmetic hyperbolic manifolds,'' \emph{Journal of Mathematical Physics},
  vol.~55, no.~8, p. 082202, 2014. [Online]. Available:
  \url{http://scitation.aip.org/content/aip/journal/jmp/55/8/10.1063/1.4891487}
\BIBentrySTDinterwordspacing

\bibitem{Freedman-Meyer-Luo-2002}
\BIBentryALTinterwordspacing
M.~H. Freedman, D.~A. Meyer, and F.~Luo, ``{$Z_2$}-systolic freedom and quantum
  codes,'' in \emph{Computational Mathematics}.\hskip 1em plus 0.5em minus
  0.4em\relax Chapman and Hall/CRC, Feb. 2002, pp. 287--320. [Online].
  Available: \url{http://dx.doi.org/10.1201/9781420035377.ch12}
\BIBentrySTDinterwordspacing

\bibitem{Vardy-1997}
A.~Vardy, ``The intractability of computing the minimum distance of a code,''
  \emph{IEEE Trans. Inform. Theory}, vol.~43, no.~6, p. 1757–1766, 1997.

\bibitem{Dumer-2003}
I.~Dumer, D.~Micciancio, and M.~Sudan, ``Hardness of approximating the minimum
  distance of a linear code,'' \emph{IEEE Trans. Inform. Theory}, vol.~49,
  no.~1, pp. 22--37, 2003.

\bibitem{Cheng-2009}
Q.~Cheng and D.~Wan, ``A deterministic reduction for the gap minimum distance
  problem,'' in \emph{STOC 2009}, 2009, pp. 33--38.

\bibitem{Gallager}
R.~G. Gallager, \emph{Low Density Parity Check Codes}.\hskip 1em plus 0.5em
  minus 0.4em\relax Cambridge, MA: M.I.T Press, 1963.

\bibitem{Litsyn-2002}
S.~Litsyn and V.~Shevelev, ``On ensembles of low-density parity-check codes:
  asymptotic distance distributions,'' \emph{IEEE Trans. Inf. Theory}, vol.~48,
  no.~4, pp. 887--908, Apr 2002.

\bibitem{Declercq-Fossorier-2008}
D.~Declercq and M.~Fossorier, ``Improved impulse method to evaluate the low
  weight profile of sparse binary linear codes,'' in \emph{2008 IEEE Intern.
  Symposium on Info. Theory}, July 2008, pp. 1963--1967.

\bibitem{Hu-Fossorier-Eleftheriou-ISIT-2004}
X.-Y. Hu, M.~P.~C. Fossorier, and E.~Eleftheriou, ``Approximate algorithms for
  computing the minimum distance of low-density parity-check codes,'' in
  \emph{2004 IEEE Intern. Symposium on Info. Theory}, June 2004.

\bibitem{Calderbank-1997}
\BIBentryALTinterwordspacing
A.~R. Calderbank, E.~M. Rains, P.~M. Shor, and N.~J.~A. Sloane, ``Quantum error
  correction via codes over {GF(4)},'' \emph{IEEE Trans. Info. Theory},
  vol.~44, pp. 1369--1387, 1998. [Online]. Available:
  \url{http://dx.doi.org/10.1109/18.681315}
\BIBentrySTDinterwordspacing

\bibitem{Gallager-1962}
R.~Gallager, ``Low-density parity-check codes,'' \emph{IRE Trans. Inf. Theory},
  vol.~8, no.~1, pp. 21--28, Jan 1962.

\bibitem{MacKay-book-2003}
\BIBentryALTinterwordspacing
D.~J.~C. MacKay, \emph{Information Theory, Inference, and Learning
  Algorithms}.\hskip 1em plus 0.5em minus 0.4em\relax New York, NY, USA:
  Cambridge University Press, 2003. [Online]. Available:
  \url{http://www.cs.toronto.edu/~mackay/itila/p0.html}
\BIBentrySTDinterwordspacing

\bibitem{Poulin-Chung-2008}
D.~Poulin and Y.~Chung, ``On the iterative decoding of sparse quantum codes,''
  \emph{Quant. Info. and Comp.}, vol.~8, p. 987, 2008.

\bibitem{Kasai-Hagiwara-Imai-Sakaniwa-2012}
K.~Kasai, M.~Hagiwara, H.~Imai, and K.~Sakaniwa, ``Quantum error correction
  beyond the bounded distance decoding limit,'' \emph{IEEE Trans. Inf. Theory},
  vol.~58, no.~2, pp. 1223 --1230, Feb 2012.

\bibitem{Barg-1998}
A.~Barg, ``Complexity issues in coding theory,'' in \emph{Handbook of Coding
  Theory}, V.~Pless and W.~C. Huffman, Eds.\hskip 1em plus 0.5em minus
  0.4em\relax Amsterdam: Elsevier, 1998, pp. 649--754.

\bibitem{Evseev-1983}
\BIBentryALTinterwordspacing
G.~S. Evseev, ``Complexity of decoding for linear codes.'' \emph{Probl.
  Peredachi Informacii}, vol.~19, no.~1, pp. 3--8, 1983, (In Russian).
  [Online]. Available: \url{http://mi.mathnet.ru/eng/ppi1159}
\BIBentrySTDinterwordspacing

\bibitem{Ekert-Macchiavello-1996}
\BIBentryALTinterwordspacing
A.~Ekert and C.~Macchiavello, ``Quantum error correction for communication,''
  \emph{Phys. Rev. Lett.}, vol.~77, pp. 2585--2588, Sep 1996. [Online].
  Available: \url{http://link.aps.org/doi/10.1103/PhysRevLett.77.2585}
\BIBentrySTDinterwordspacing

\bibitem{Feng-Ma-2004}
\BIBentryALTinterwordspacing
K.~Feng and Z.~Ma, ``A finite {G}ilbert-{V}arshamov bound for pure stabilizer
  quantum codes,'' \emph{Information Theory, IEEE Transactions on}, vol.~50,
  no.~12, pp. 3323--25, dec. 2004. [Online]. Available:
  \url{http://dx.doi.org/10.1109/TIT.2004.838088}
\BIBentrySTDinterwordspacing

\bibitem{White-Grassl-2006}
G.~White and M.~Grassl, ``A new minimum weight algorithm for additive codes,''
  in \emph{2006 IEEE Intern. Symp. Inform. Theory}, July 2006, pp. 1119--1123.

\bibitem{Dumer-1986}
I.~I. Dumer, ``On syndrome decoding of linear codes,'' in \emph{Proc. Ninth
  All-Union Symp. Redundancy in Information Systems}.\hskip 1em plus 0.5em
  minus 0.4em\relax Nauka, May 1986, vol.~2, pp. 157--159, (In Russian).

\bibitem{Dumer-1989}
\BIBentryALTinterwordspacing
------, ``Two decoding algorithms for linear codes,'' \emph{Probl. Peredachi
  Informacii}, vol.~25, no.~1, pp. 24--32, 1989, (In Russian). [Online].
  Available: \url{http://mi.mathnet.ru/eng/ppi635}
\BIBentrySTDinterwordspacing

\bibitem{Dumer-2001}
I.~Dumer, ``Soft-decision decoding using punctured codes,'' \emph{IEEE Trans.
  Inf. Theory}, vol.~47, no.~1, pp. 59--71, Jan 2001.

\bibitem{Calderbank-Shor-1996}
A.~R. Calderbank and P.~W. Shor, ``Good quantum error-correcting codes exist,''
  \emph{Phys. Rev. A}, vol.~54, no.~2, pp. 1098--1105, Aug 1996.

\bibitem{Prange-1962}
E.~Prange, ``The use of information sets in decoding cyclic codes,''
  \emph{Information Theory, IRE Transactions on}, vol.~8, no.~5, pp. 5--9,
  1962.

\bibitem{McEliece-1978}
R.~McEliece, ``A public-key cryptosystem based on algebraic coding theory,''
  JPL, Tech. Rep. DSN Progress Report 43-44, 1978.

\bibitem{Lee-1988}
P.~J. Lee and E.~F. Brickell, ``An observation on the security of mceliece’s
  public-key cryptosystem,'' in \emph{Advances in Cryptology - EUROCRYPT 1988},
  1988, pp. 275--280.

\bibitem{Leon-1988}
J.~S. Leon, ``A probabilistic algorithm for computing minimum weights of large
  error-correcting codes,'' \emph{IEEE Trans. Info. Theory}, vol.~34, no.~5,
  pp. 1354 --1359, Sep 1988.

\bibitem{Kruk-1989}
\BIBentryALTinterwordspacing
E.~A. Kruk, ``Decoding complexity bound for linear block codes,'' \emph{Probl.
  Peredachi Inf.}, vol.~25, no.~3, pp. 103--107, 1989, (In Russian). [Online].
  Available: \url{http://mi.mathnet.ru/eng/ppi665}
\BIBentrySTDinterwordspacing

\bibitem{Coffey-Goodman-1990}
J.~T. Coffey and R.~M. Goodman, ``The complexity of information set decoding,''
  \emph{IEEE Trans. Info. Theory}, vol.~36, no.~5, pp. 1031 --1037, Sep 1990.

\bibitem{Erdos-book}
P.~Erdos and J.~Spencer, \emph{Probabilistic methods in combinatorics}.\hskip
  1em plus 0.5em minus 0.4em\relax Budapest: Akademiai Kiado, 1974.

\bibitem{Rich-Urb-2001}
T.~Richardson and R.~Urbanke, ``The capacity of low-density parity check codes
  under message passing decoding,'' \emph{IEEE Trans. Inf. Theory}, vol.~47,
  no.~2, pp. 599--618, February 2001.

\bibitem{Pishro-2004}
H.~Pishro-Nik and F.~Fekri, ``On decoding of low-density parity-check codes
  over the binary erasure channel,'' \emph{IEEE Trans. Inf. Theory}, vol.~50,
  no.~3, pp. 439--454, March 2004.

\bibitem{Hsu-2008}
C.-H. Hsu and A.~Anastasopoulosi, ``Capacity achieving ldpc codes through
  puncturing,'' \emph{IEEE Trans. Inf. Theory}, vol.~54, no.~10, pp.
  4698--4706, Oct. 2008.

\bibitem{Luby-2001}
M.~Luby, M.~Mitzenmacher, A.~Shokrollahi, and D.~Spielmani, ``Efficient erasure
  correcting codesl,'' \emph{IEEE Trans. Inf. Theory}, vol.~47, no.~2, pp.
  569--584, February 2001.

\bibitem{Di-2001}
C.~Di, R.~Urbanke, and T.~Richardson, ``Weight distributions: How deviant can
  you be?'' in \emph{Proc. Int. Symp. Information Theory (ISIT 2001),
  Washington, DC}, 2001, p.~50.

\bibitem{Litsyn-2003}
S.~Litsyn and V.~Shevelev, ``Distance distributions in ensembles of irregular
  low-density parity-check codes,'' \emph{IEEE Trans. Inf. Theory}, vol.~49,
  no.~12, pp. 3140--3159, Dec 2003.

\bibitem{Til-2009}
V.~R. l.~Andriyanova and J.-P. Tillich, ``Binary weight distribution of
  non-binary ldpc codes,'' in \emph{Proc. Int. Symp. Information Theory (ISIT
  2009), Seoul, Korea}, 2009, pp. 65--69.

\bibitem{Yang-2011}
Y.~C. Z.~Z. S.~Yang, T.~Honold and P.~Qiu, ``Weight distributions of regular
  low-density parity-check codes over finite fields,'' \emph{IEEE Trans. Inf.
  Theory}, vol.~57, no.~11, pp. 7507--7521, Nov. 2011.

\bibitem{Stern-1989}
J.~Stern, ``A method for finding codewords of small weight,'' in \emph{Coding
  Theory and Applications}, ser. LNCS, G.~Cohen and J.~Wolfmann, Eds.\hskip 1em
  plus 0.5em minus 0.4em\relax Heidelberg: Springer, 1989, vol. 388, pp.
  106--113.

\bibitem{Dumer-1991}
I.~Dumer, ``On minimum distance decoding of linear codes,'' in \emph{Fifth
  Soviet-Swedish intern. workshop Information theory}, G.~Kabatianskii,
  Ed.\hskip 1em plus 0.5em minus 0.4em\relax Moscow: Nauka, Jan. 1991, pp.
  50--52.

\bibitem{Sendrier-2009}
M.~Finiasz and N.~Sendrier, ``Security bounds for the design of code-based
  cryptosystems,'' in \emph{Asiacrypt 2009}, ser. LNCS.\hskip 1em plus 0.5em
  minus 0.4em\relax Heidelberg: Springer, 2011, vol. 5912, pp. 88--105.

\bibitem{Bernstein-2011}
D.~J. Bernstein, T.~Lange, and C.~Peters, ``Smaller decoding exponents:
  ball-collision decoding,'' in \emph{CRYPTO 2011}, ser. LNCS.\hskip 1em plus
  0.5em minus 0.4em\relax Heidelberg: Springer, 2011, vol. 6841, pp. 743--760.

\bibitem{Becker-2012}
A.~Becker, A.~Joux, A.~May, and A.~Meurer, ``Decoding random binary linear
  codes in 2 n/20 : How $1 + 1 = 0$ improves information set decoding,'' in
  \emph{EUROCRYPT 2012}, ser. LNCS.\hskip 1em plus 0.5em minus 0.4em\relax
  Heidelberg: Springer, 2012, vol. 7237, pp. 520--536.

\end{thebibliography}

\end{document}